\documentclass[journal]{IEEEtran}
\usepackage{graphicx}
\usepackage{amsfonts}
\usepackage{amsmath}
\usepackage{amssymb}
\usepackage{mathrsfs}
\usepackage{bm}
\usepackage{color}
\usepackage{array}
\usepackage{booktabs}
\usepackage{multirow}
\newtheorem{thm}{Theorem}
\newtheorem{cor}{Corollary}
\newtheorem{lem}{Lemma}
\newtheorem{proof}{proof}

\newtheorem{defn}{Definition}
\newtheorem{rem}{Remark}
\newtheorem{exam}{Example}

\begin{document}

\title{Sequence-Subset Distance and Coding for Error Control in
DNA-based Data Storage}

\author{Wentu~Song,
        Kui~Cai,~\IEEEmembership{senior member,~IEEE}
        and~Kees~A.~Schouhamer~Immink,~\IEEEmembership{Fellow,~IEEE}
\thanks{Wentu Song and Kui Cai are with Singapore University of
        Technology and Design, Singapore, e-mail:
        \{wentu$\_$song, cai$\_$kui\}@sutd.edu.sg.}
\thanks{Kees A. Schouhamer Immink is with Turing Machines Inc,
        Willemskade 15d, 3016 DK Rotterdam, The Netherlands,
        e-mails: immink@turing-machines.com.}
        }


\maketitle

\begin{abstract}
The process of DNA-based data storage (DNA storage for short) can
be mathematically modelled as a communication channel, termed DNA
storage channel, whose inputs and outputs are sets of unordered
sequences. To design error correcting codes for DNA storage
channel, a new metric, termed the \emph{sequence-subset distance},
is introduced, which generalizes the Hamming distance to a
distance function defined between any two sets of unordered
vectors and helps to establish a uniform framework to design error
correcting codes for DNA storage channel. We further introduce a
family of error correcting codes, referred to as
\emph{sequence-subset codes}, for DNA storage and show that the
error-correcting ability of such codes is completely determined by
their minimum distance. We derive some upper bounds on the size of
the sequence-subset codes including a tight bound for a special
case, a Singleton-like bound and a Plotkin-like bound. We also
propose some constructions, including an optimal construction for
that special case, which imply lower bounds on the size of such
codes.
\end{abstract}


\IEEEpeerreviewmaketitle

\section{Introduction}
The idea of storing data in synthetic DNA strands (sequences) has
been around since $1988$ \cite{Davis96} and DNA-based data storage
has been progressing rapidly in recent years with the development
of DNA synthesis and sequencing technology. Compared to
traditional magnetic and optical media, DNA storage has some
competing advantages such as extreme high density, long durability
\cite{Bornholt16}, and low energy consumption \cite{Church12}.

A DNA strand is mathematically represented by a quaternary
sequence, each symbol represents one of the four types of base
nucleotides: adenine $(${\small\textsf{A}}$)$, cytosine
$(${\small\textsf{C}}$)$, guanine $(${\small\textsf{G}}$)$ and
thymine $(${\small\textsf{T}}$)$. Basically, in a DNA-based
storage system, the original binary data is first encoded to a set
of quaternary sequences. Then the corresponding DNA strands are
synthesized and stored in DNA pools. To retrieve (read) the
original data, the stored DNA strands are sequenced to generate a
set of quaternary sequences, which are then decoded to the
original binary data. The process of DNA synthesizing, storing and
sequencing can be mathematically modelled as a communication
channel, called the DNA storage channel, which can be depicted by
Fig. \ref{Systm-model}.

\renewcommand\figurename{Fig}
\begin{figure*}[htbp]
\begin{center}
\includegraphics[height=2.3cm]{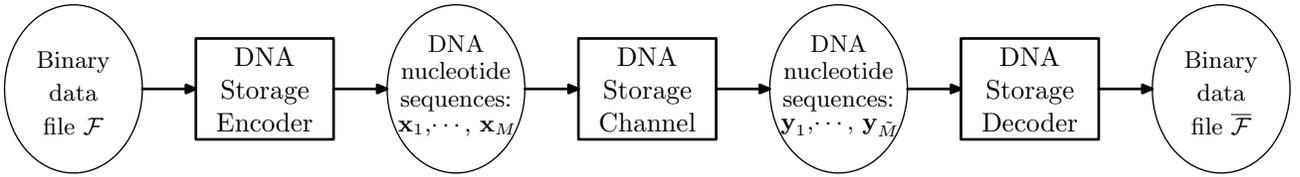}
\end{center}
\vspace{-0.2cm}\caption{System model of the DNA storage: The DNA
storage channel is the mathematical model of the process of DNA
synthesizing, storing and sequencing. A reliable system should
guarantee that with sufficiently high probability the decoded file
$\overline{\mathcal F}$ equals to the original file $\mathcal
F$.}\label{Systm-model}
\end{figure*}

The sequencing process can be modelled as a randomly sampling and
reading of molecules with replacement from the DNA pool
\cite{Heckel}. Some DNA strands may have many copies that are
sequenced while some strands may never be sequenced. Moreover,
since the synthesis/sequencing process is prone to errors, a
specific DNA strand in the pool may have many noisy copies that
are contained in the sequencing output. These sequenced strands
are clustered according to their Levenshtein distance or by some
other methods $($e.g., see \cite{Organick18,Rashtchian17}$)$, and
then the clustered sequences are reconstructed by performing an
estimate for each cluster \cite{Levenshtein01}. All different
estimated sequences form the output of the DNA storage channel.
Another characteristic of DNA storage channel is that unlike the
conventional magnetic or optical recording systems, the DNA
sequences are stored in ``pools'', where structured addressing is
not allowed. Therefore, the inputs and outputs of the DNA storage
channel can be viewed as sets of unordered DNA sequences.

The output of the DNA storage channel may be distorted by the
following five types of errors:
\begin{itemize}
 \item[$\bullet$] \emph{Sequence deletion}: One or more of the input
 sequences are lost. A DNA strand is lost if
 it is never sequenced. Another case of sequence deletion is
 when there are $t~(>1)$ strands
 that are changed to the same strand by substitution errors, then any
 $t-1$ of them are viewed as lost sequences.
 As a result, the number of output sequences is
 smaller than the number of input sequences.
 \item[$\bullet$] \emph{Sequence insertion}: One or more sequences
 that do not belong to the set of input sequences are added into the
 output sequences.
 If the output of the channel contains $t~(>1)$ different
 noisy copies of an input sequence, then any $t-1$ of them can be viewed
 as inserted sequences.
 As a result, the number of output sequences is
 larger than the number of input sequences.
 \item[$\bullet$] \emph{Symbol deletion}: One or more
 symbols in a sequence are removed. As a result, the length of the
 erroneous sequence is decreased.
 \item[$\bullet$] \emph{Symbol insertion}: One or more symbols are
 added into a sequence. As a result, the length of the erroneous
 sequence is increased.
 \item[$\bullet$] \emph{Symbol substitution}: One or more symbols
 in a sequence are replaced by other symbols. In this case, the
 length of the erroneous sequence remains unchanged.
\end{itemize}
Note that sequence deletion and sequence insertion can take place
simultaneously. If the number of sequence deletions equals the
number of sequence insertions, then the total number of input
sequences remain unchanged.

To combat different types of errors in DNA synthesizing and
sequencing, various coding techniques are used by DNA storage.
Most demonstration research works employ constrained coding
combined with classical error correcting codes (e.g. Reed-Solomon
codes) \cite{Church12}-\cite{Immink18}. In addition, to combat the
lack of ordering of the transmitted sequences, a unique address
(index) is added to each sequence.

Codes that can correct $s$ (or fewer) losses of sequences and $e$
(or fewer) substitutions in each of $t$ (or fewer) sequences were
studied in \cite{Lenz18} by considering the so-called error ball.
Codes dealing with insertion/deletion errors were also studied in
\cite{Lenz18}. Codes that can correct a total of $K$ substitution
errors were studied in \cite{Sima18}.

\subsection{Our Contribution}
In this paper, we consider error control for DNA storage channel
by introducing a new metric, termed the \emph{sequence-subset
distance}, over the power set of the set of all vectors of fixed
length over a finite alphabet, which is the space of the
inputs/outputs of the DNA storage channel. This metric is a
generalization of the classical Hamming distance and can help to
establish a uniform framework to design codes for DNA storage
channel that can correct errors of sequence deletion, sequence
insertion and symbol substitution.

We study error correcting codes with respect to the
sequence-subset distance, which we refer to as
\emph{sequence-subset codes}, for DNA-based data storage. We show
that similar to codes with respect to the classical Hamming
distance, a sequence-subset code $\mathcal C$ can correct any
number of $n_{\text{D}}$ sequence deletions, $n_{\text{I}}$
sequence insertions, and totally $n_{\text{S}}$ symbol
substitutions, provided that
$n_{\text{S}}+L\cdot\max\{n_{\text{I}},n_{\text{D}}\}\leq
\frac{d_{\text{S}}(\mathcal C)-1}{2}$, where $L$ is the length of
the sequences and $d_{\text{S}}(\mathcal C)$ is the minimum
distance of $\mathcal C$.

We derive some upper bounds on the size of the sequence-subset
codes including a tight bound for the special case that $d=LM$, a
Singleton-like bound and a Plotkin-like bound, where $M$ is the
codeword size (i.e., the number of sequences in each codeword of
the sequence-subset codes).

We give a construction of optimal codes (with respect to size) for
the special case that $d=LM$ and $M^{\frac{1}{L}}$ is an integer,
where $d$ is the minimum distance of the code. We also give some
general constructions of sequence-subset codes, which imply lower
bounds of the size of such codes.

\subsection{Related Work}

The similar channel model for DNA storage was also studied in
\cite{Lenz18}, \cite{Sima18} and \cite{Heckel}.

In \cite{Lenz18} and \cite{Sima18}, data is stored in an unordered
set of $M$ strings of length $L$ (the input of the DNA storage
channel), where $M$ and $L$ are some fixed positive integers. The
work of \cite{Lenz18} considered the error-correcting problem by
restricting that $s$ sequences are lost during the
synthesizing/sequencing process and the output of the channel is a
subset of $M-s$ input sequences, among which $M-s-t$ sequences are
correctly reconstructed and $t$ sequences are reconstructed with
errors such that each sequence has at most $\epsilon$ errors,
where possible errors are symbol insertion/deletion and/or
substitution. In \cite{Sima18}, the channel was studied under the
assumption that the values of a total of $K$ different positions
in the $M$ input sequences are changed (i.e., there are totally
$K$ symbol substitutions). Since the erroneous sequence may be
equal to another existing sequence, which in fact induces sequence
deletion, the output of the channel is a set of $T$ strings of
length $L$ for some $T$ such that $M-K\leq T\leq M$.

In \cite{Heckel}, DNA storage is modelled as a channel whose
inputs are multisets of $M$ sequences of length $L$ while the
output of the channel is a multiset of $N$ sequences of length
$L$, which is obtained by drawing $N$ samples independently and
uniformly at random, with replacement, from the $M$ input
sequences, where $M,L,N$ are the fixed parameters of the channel.
It also assumes that each sampled molecule is read error-free.

Comparison of our model with the models of \cite{Lenz18},
\cite{Sima18} and \cite{Heckel} is given in Table 1.

\begin{table*}
\begin{center}\renewcommand\arraystretch{1.3}
\begin{tabular}{|p{2.2cm}|p{4.3cm}|p{4.3cm}|p{4.2cm}|}
\hline ~~~~~Model & ~~~~~~~~~~~~Inputs
& ~~~~~~~~~~~~~Outputs & ~~~~~~~~~~~Error types\\
\hline Model of \cite{Lenz18} & A set of $M$ DNA sequences of
length $L$, for some fixed $M$ and $L$ & A set of $M-s$ DNA
sequences with at most $t$ erroneous sequences & Sequence
deletion, symbol insertion, deletion and substitution \\
\hline Model of \cite{Sima18} & A set of $M$ DNA sequences of
length $L$, for some fixed $M$ and $L$ & A set of $T$ DNA
sequences with totally $K$ substitutions, where $M-K\leq T\leq M$
& Symbol substitution, sequence deletion induced by
symbol substitution \\
\hline Model of \cite{Heckel} & A multi-set of $M$ DNA sequences
of length $L$ for some fixed $M$ and $L$ & A multi-set of $N$ DNA
sequences drawn randomly with replacement from the $M$ input
sequences & Sequence deletion \\
\hline Our Model & A set of $M$ DNA sequences of length $L$, where
$L$ is fixed but $M$ is not necessarily fixed & A set of
$\tilde{M}$ DNA sequences, which may include erroneous sequences
and additional inserted sequences & Sequence
deletion, sequence insertion and symbol substitution \\
\hline
\end{tabular}\\
\vspace{2.5mm} Table 1. Comparison of different models for DNA
storage channel.
\end{center}
\end{table*}

\vspace{2.5mm} Another model for DNA storage channel, which
focuses on modelling the process of synthesis and sequencing of
single DNA strand, was consider in \cite{Kiah-IT16}. Different
from our model, the input of this channel is a single DNA sequence
(rather than a set of sequences), and through the process of
synthesis and sequencing, a set of DNA fragments along with their
frequency count is obtained, which can be represented by a profile
vector. Three types of errors, namely, substitution errors due to
synthesis, coverage errors, and $\ell$-gram substitution errors
due to sequencing, are considered in \cite{Kiah-IT16}.

There are still some other communication channels similar to DNA
storage channel. The permutation channel considered in
\cite{Langberg-IT} has input and output as vectors over a finite
alphabet and the transmitted vector is corrupted by a permutation
on its coordination. The permutation channel with impairments was
considered in \cite{Kovacevic17}, where the input and output are
multi-sets, rather than vectors, of symbols from a finite
alphabet. Unlike the DNA storage channel, the structure
information of the sequences (i.e., the Hamming distance between
sequences when the sequence length $L>1$) is not considered in
such models.

\subsection{Organization}
The rest of the paper is organized as follows. In Section
\uppercase\expandafter{\romannumeral 2}, we introduce the
sequence-subset distance and provide the basic properties of codes
with sequence-subset distance. We analyze the upper bound on the
size of sequence-subset codes in Section
\uppercase\expandafter{\romannumeral 3} and give some
constructions of such codes in Section
\uppercase\expandafter{\romannumeral 4}. The paper is concluded in
Section \uppercase\expandafter{\romannumeral 5}.

\subsection{Notations}
The following notations will be used in this paper:
\begin{itemize}
 \item [1)] For any positive integer $n$, $[n]:=\{1,2,\cdots,n\}$.
 \item [2)] For any set $\mathbb A$, $|\mathbb A|$ denotes the size
 (i.e., cardinality) of $\mathbb A$ and $\mathcal
 P(\mathbb A)$ denotes the power set of $\mathbb A~($i.e.,
 the collection of all subsets of $\mathbb A)$.
 \item [3)] For any two sets $\textbf{\text{X}}$ and $\textbf{\text{Y}}$,
 $\textbf{\text{X}}\backslash\textbf{\text{Y}}$ is the set of all
 elements of $\textbf{\text{X}}$ that do not belong to
 $\textbf{\text{Y}}$.
 \item [4)] For any $n$-tuple
 $\textbf{\text{x}}\in\mathbb A^n$ and any $i\in[n]$,
 $\textbf{\text{x}}(i)$ denotes the $i$th coordinate of
 $\textbf{\text{x}}$, and hence $\textbf{\text{x}}$ is denoted as
 $\textbf{\text{x}}=(\textbf{\text{x}}(1), \textbf{\text{x}}(2),
 \cdots, \textbf{\text{x}}(n))$.
\end{itemize}

\section{Preliminary}
We first introduce the concept of sequence-subset distance. Then
we discuss the error pattern and error-correcting in DNA storage
channel using codes with sequence-subset distance.

\subsection{Sequence-Subset Distance}
Let $\mathbb A$ be a fixed finite alphabet. For DNA data storage,
typically $\mathbb A=\{${\small\textsf{A}}, {\small\textsf{T}},
{\small\textsf{C}}, {\small\textsf{G}}$\}$, representing the four
types of base nucleotides. In this work, for generality, we assume
that $\mathbb A$ is any fixed finite alphabet of size $q\geq 2$.

Let $L$ be a positive integer. For any
$\textbf{x}_1,\textbf{x}_2\in\mathbb A^L$, the Hamming distance
between $\textbf{x}_1$ and $\textbf{x}_2$, denoted by
$d_{\text{H}}(\textbf{x}_1,\textbf{x}_2)$, is defined as the
number of coordinates where $\textbf{x}_1$ and $\textbf{x}_2$
differ, that is,
$$d_{\text{H}}(\textbf{x}_1,\textbf{x}_2):=|\{i\in[L];
\textbf{x}_1(i)\neq \textbf{x}_2(i)\}|.$$

For any two subsets $\textbf{\text{X}}_1$ and
$\textbf{\text{X}}_2$ of $\mathbb A^L$ such that
$|\textbf{\text{X}}_1|\leq|\textbf{\text{X}}_2|$ and any injection
$\chi:\textbf{\text{X}}_1\rightarrow\textbf{\text{X}}_2$, denote
\begin{align}\label{eq-def-chi-d}
d_{\chi}(\textbf{\text{X}}_1,\textbf{\text{X}}_2)\!:=\!\!
\sum_{\textbf{\text{x}}\in\textbf{\text{X}}_1}d_{\text{H}}(\textbf{\text{x}},
\chi(\textbf{\text{x}}))\!+\!L(|\textbf{\text{X}}_2|\!-
\!|\textbf{\text{X}}_1|).\end{align} Then a natural way to
generalize the Hamming distance to the space of all subsets of
$\mathbb A^L$ is as follows.
\begin{defn}\label{def-dst}
For any $\textbf{\text{X}}_1,\textbf{\text{X}}_2\subseteq\mathbb
A^L$, without loss of generality, assuming
$|\textbf{\text{X}}_1|\leq|\textbf{\text{X}}_2|$, the
\emph{sequence-subset distance} between $\textbf{\text{X}}_1$ and
$\textbf{\text{X}}_2$ is defined as
\begin{align}\label{eq-def-dst}
d_{\text{S}}(\textbf{\text{X}}_1,\textbf{\text{X}}_2)=
d_{\text{S}}(\textbf{\text{X}}_2,\textbf{\text{X}}_1):=\min_{\chi\in\mathscr
X}d_{\chi}(\textbf{\text{X}}_1,\textbf{\text{X}}_2),\end{align}
where $\mathscr X$ is the set of all injections
$\chi:\textbf{\text{X}}_1\rightarrow\textbf{\text{X}}_2$.\footnote{A
more accurate notation for the set $\mathscr X$ is $\mathscr
X_{\textbf{\text{X}}_1,\textbf{\text{X}}_2}$ because it is related
to the subsets $\textbf{\text{X}}_1$ and $\textbf{\text{X}}_2$.
However, we can omit the subscripts safely because they can be
easily specified by the context.}
\end{defn}

\begin{exam}
Suppose $\mathbb A=\{0,1\}$ and $L=4$. Consider
$\textbf{X}_1=\{\textbf{x}_1,\textbf{x}_2,\textbf{x}_3\}$ and
$\textbf{X}_2=\{\textbf{y}_1,\textbf{y}_2,\textbf{y}_3,\textbf{y}_4\}$,
where $\textbf{x}_1=1010$, $\textbf{x}_2=0010$,
$\textbf{x}_3=1101$, $\textbf{y}_1=1101$, $\textbf{y}_2=0011$,
$\textbf{y}_3=1011$ and $\textbf{y}_4=1100$. Let
$\chi_0:\textbf{X}_1\rightarrow\textbf{X}_2$ be such that
$\chi_0(\textbf{x}_1)=\textbf{y}_3$,
$\chi_0(\textbf{x}_2)=\textbf{y}_2$ and
$\chi_0(\textbf{x}_3)=\textbf{y}_1$. Then by \eqref{eq-def-chi-d},
we can obtain $d_{\chi_0}(\textbf{X}_1,\textbf{X}_2)=6$. We can
further verify that $d_{\chi}(\textbf{X}_1,\textbf{X}_2)\geq 6$
for all injections $\chi: \textbf{X}_1\rightarrow \textbf{X}_2$.
Hence by \eqref{eq-def-dst}, we have
$d_{\text{S}}(\textbf{X}_1,\textbf{X}_2)
=d_{\chi_0}(\textbf{X}_1,\textbf{X}_2)=6$.
\end{exam}

\begin{rem}
Given any subsets $\textbf{X}_1,\textbf{X}_2$ of
 $\mathbb A^L$ such that
$|\textbf{X}_1|\leq|\textbf{X}_2|$, let $V_1=\textbf{X}_1\cup V_0$
and $V_2=\textbf{X}_2$, where $V_0$ is a set disjoint with
$\textbf{X}_1$. We can construct a complete bipartite weighted
graph $G$ with bipartition $(V_1, V_2)$ such that for each
$\textbf{x}_1\in V_1$ and $\textbf{x}_2\in V_2$, the weight of the
edge $(\textbf{x}_1,\textbf{x}_2)$ is
$d_{\text{H}}(\textbf{x}_1,\textbf{x}_2)$, where we define
$d_{\text{H}}(\textbf{x}_1,\textbf{x}_2)=L$ for any
$\textbf{x}_1\in V_0$ and $\textbf{x}_2\in V_2$. Then by
Definition \ref{def-dst}, the sequence-subset distance between
$\textbf{X}_1$ and $\textbf{X}_2$ can be computed from a minimum
weight perfect matching of $G$, which can be done in time
$O(|\textbf{X}_2|^3)$ using the Kuhn-Munkres algorithm
\cite{Munkres}.
\end{rem}

We now prove some important properties of the function
$d_{\text{S}}(\cdot,\cdot)$ and then prove that it is really a
distance function.

First, intuitively, the elements in
$\textbf{\text{X}}_1\cap\textbf{\text{X}}_2$ should have no effect
on the sequence-subset distance between $\textbf{\text{X}}_1$ and
$\textbf{\text{X}}_2$. This is shown to be true by the following
lemma and corollary.
\begin{lem}\label{lem-dst}
For any $\textbf{\text{X}}_1,\textbf{\text{X}}_2\subseteq\mathbb
A^L$ such that $|\textbf{\text{X}}_1|\leq|\textbf{\text{X}}_2|$,
there exists an injection $\chi_0\in\mathscr X$ such that
$d_{\text{S}}(\textbf{\text{X}}_1,\textbf{\text{X}}_2)=
d_{\chi_0}(\textbf{\text{X}}_1,\textbf{\text{X}}_2)$ and
$\chi_0(\textbf{\text{x}})=\textbf{\text{x}}$ for all
$\textbf{\text{x}}\in\textbf{\text{X}}_1\cap\textbf{\text{X}}_2$.
\end{lem}
\begin{proof}
The proof is given in Appendix A.
\end{proof}

\begin{cor}\label{cor-dst}
For any two subsets $\textbf{\text{X}}_1$ and
$\textbf{\text{X}}_2$ of $\mathbb A^L$,
\begin{align*}
d_{\text{S}}(\textbf{\text{X}}_1,\textbf{\text{X}}_2)=
d_{\text{S}}(\textbf{\text{X}}_1\backslash\textbf{\text{X}}_2,
\textbf{\text{X}}_2\backslash\textbf{\text{X}}_1).\end{align*}
\end{cor}
\begin{proof}
This corollary is just a direct consequence of Definition
\ref{def-dst} and Lemma \ref{lem-dst}.
\end{proof}

\begin{lem}\label{lem-dst-subset}
Suppose $\textbf{\text{X}}_1,\textbf{\text{X}}_2\subseteq\mathbb
A^L$ such that $|\textbf{\text{X}}_1|\leq|\textbf{\text{X}}_2|$.
Suppose $\textbf{X}_2'\subseteq \textbf{X}_2$ such that
$|\textbf{X}_1|\leq|\textbf{X}_2'|$. Then
$$d_{\text{S}}(\textbf{\text{X}}_1,\textbf{\text{X}}_2')\leq
d_{\text{S}}(\textbf{\text{X}}_1,\textbf{\text{X}}_2).$$
\end{lem}
\begin{proof}
The proof is given in Appendix B.
\end{proof}

Now we prove that $d_{\text{S}}(\cdot,\cdot)$ is really a distance
function (metric) over $\mathcal P(\mathbb A^L)$.
\begin{thm}\label{dst-mtrc} The function
$d_{\text{S}}(\cdot,\cdot)$ is a distance function over the power
set $\mathcal P(\mathbb A^L)$.
\end{thm}
\begin{proof}
The proof is given in Appendix C.
\end{proof}

\subsection{Error Pattern of DNA Storage Channel}
In this paper we consider DNA storage channel with sequence
deletion/insertion and symbol substitution. The input of the
channel is a set of unordered sequences
$$\textbf{\text{X}}=\{\textbf{\text{x}}_1,\textbf{\text{x}}_2,\cdots,
\textbf{\text{x}}_M\}\subseteq\mathbb A^L$$ and the output is
another set of unordered sequences
$$\textbf{\text{Y}}=\{\textbf{\text{y}}_1,\textbf{\text{y}}_2,\cdots,
\textbf{\text{y}}_{\tilde{M}}\}\subseteq\mathbb A^L,$$ where $L$
is the length of the sequences. Usually,
$\textbf{\text{Y}}\neq\textbf{\text{X}}$ because of the channel
noise. Sequences in the subset
$\textbf{\text{X}}\cap\textbf{\text{Y}}$ are correctly
transmitted; Sequences in $\textbf{X}\backslash\textbf{Y}$ are
either lost (sequence deletion) or changed to sequences in
$\textbf{Y}\backslash\textbf{X}$ (symbol substitution); Sequences
in $\textbf{Y}\backslash\textbf{X}$ are either excessive (sequence
insertion) or obtained from some sequences in
$\textbf{X}\backslash\textbf{Y}$ (symbol substitution). Let
$n_{\text{I}}$, $n_{\text{D}}$ and $n_{\text{S}}$ denote the total
number of sequence insertions, sequence deletions and symbol
substitutions, respectively, in $\textbf{Y}$. Then we call the
$3$-tuple $(n_{\text{I}}, n_{\text{D}}, n_{\text{S}})$ the
\emph{error pattern} of $\textbf{Y}$. Furthermore, we have the
following lemma.

\begin{lem}\label{lem-ep-dst}
Suppose the channel input is $\textbf{X}$ and output is
$\textbf{Y}$. If the error pattern of $\textbf{Y}$ is
$(n_{\text{I}}, n_{\text{D}}, n_{\text{S}})$, then
$$d_{\text{S}}(\textbf{X},\textbf{Y})\leq
n_{\text{S}}+L\cdot\max\{n_{\text{I}}, n_{\text{D}}\}.$$
\end{lem}
\begin{proof}
Note that we can always partition the two subsets
$\textbf{X}\backslash \textbf{Y}$ and $\textbf{Y}\backslash
\textbf{X}$ as
$$\textbf{X}\backslash \textbf{Y}=\textbf{X}_{\text{D}}\cup
\textbf{X}_{\text{S}}~~~\text{and}~~~\textbf{Y}\backslash
\textbf{X}=\textbf{Y}_{\text{I}}\cup \textbf{Y}_{\text{S}},$$
where $\textbf{X}_{\text{D}}$ is the set of lost input sequences,
$\textbf{X}_{\text{S}}$ is the set of input sequences that are
changed to $\textbf{Y}_{\text{S}}$ by symbol substitution, and
$\textbf{Y}_{\text{I}}$ is the set of sequences that are inserted
to $\textbf{Y}$. Clearly, we have
$$n_{\text{I}}=|\textbf{Y}_{\text{I}}|~~~\text{and}~~~
n_{\text{D}}=|\textbf{X}_{\text{D}}|.$$ Moreover,
$|\textbf{X}_{\text{S}}|=|\textbf{Y}_{\text{S}}|$ and there exists
a bijection
$\chi:\textbf{X}_{\text{S}}\rightarrow\textbf{Y}_{\text{S}}$ such
that for each $\textbf{x}\in\textbf{X}_{\text{S}}$,
$\chi(\textbf{x})$ is the erroneous sequence of $\textbf{x}$ by
symbol substitution. Hence, we have
$$n_{\text{S}}=\sum_{\textbf{x}\in\textbf{X}_{\text{S}}}
d_{\text{H}}(\textbf{x}, \chi(\textbf{x})).$$ For further
discussion, we need to consider the following two cases.

Case 1: $n_{\text{I}}\leq n_{\text{D}}$. In this case, we have
$|\textbf{Y}_{\text{I}}|=n_{\text{I}}\leq
n_{\text{D}}=|\textbf{X}_{\text{D}}|$ and
$|\textbf{Y}\backslash\textbf{X}|\leq|\textbf{X}
\backslash\textbf{Y}|$. Let
$\chi':\textbf{Y}_{\text{I}}\rightarrow\textbf{X}_{\text{D}}$ be
any fixed injection. Then we can obtain an injection
$\bar{\chi}:\textbf{Y}\backslash\textbf{X}\rightarrow\textbf{X}
\backslash\textbf{Y}$ such that
\begin{equation*}
\bar{\chi}(\textbf{y})=\left\{\begin{aligned}
&\chi^{-1}(\textbf{y})& ~ ~ _{~}\text{if}~\textbf{y}\in\textbf{Y}_{\text{S}};\\
&\chi'(\textbf{y})& ~ ~\text{if}~\textbf{y}\in\textbf{Y}_{\text{I}}.\\
\end{aligned} \right. 
\end{equation*}
Since $|\textbf{X}
\backslash\textbf{Y}|-|\textbf{Y}\backslash\textbf{X}|=|\textbf{X}_{\text{D}}|
-|\textbf{Y}_{\text{I}}|=n_{\text{D}}-n_{\text{I}}$, then by
\eqref{eq-def-chi-d},
\begin{align*}
d_{\bar{\chi}}(\textbf{Y}\backslash\textbf{X},\textbf{X}\backslash\textbf{Y})
&=\sum_{\textbf{y}\in\textbf{Y}\backslash\textbf{X}}
d_{\text{H}}(\textbf{y},\bar{\chi}(\textbf{y}))+L\cdot\left(|\textbf{X}
\backslash\textbf{Y}|-|\textbf{Y}\backslash\textbf{X}|\right)\\
&=\sum_{\textbf{y}\in\textbf{Y}_{\text{S}}}
d_{\text{H}}(\textbf{y},\chi(\textbf{y}))
+\sum_{\textbf{y}\in\textbf{Y}_{\text{I}}}
d_{\text{H}}(\textbf{y},\chi'(\textbf{y}))\\
&~~~~+L\cdot\left(n_{\text{D}}-n_{\text{I}}\right)\\
&\leq n_{\text{S}}+L\cdot n_{\text{I}}+
L\cdot (n_{\text{D}}-n_{\text{I}})\\ &= n_{\text{S}}+L\cdot n_{\text{D}}\\
&=n_{\text{S}}+L\cdot\max\{n_{\text{I}},
n_{\text{D}}\}\end{align*} where the inequality comes from the
simple fact that $d_{\text{H}}(\textbf{z},\textbf{z}')\leq L$ for
any $\textbf{z},\textbf{z}'\in\mathbb A^L$. Hence, by Corollary
\ref{cor-dst} and Definition \ref{def-dst}, we have
\begin{align*}
d_{\text{S}}(\textbf{X},\textbf{Y})
&=d_{\text{S}}(\textbf{X}\backslash\textbf{Y},
\textbf{Y}\backslash\textbf{X})\\&\leq
d_{\bar{\chi}}(\textbf{Y}\backslash\textbf{X},
\textbf{X}\backslash\textbf{Y})\\
&\leq n_{\text{S}}+L\cdot\max\{n_{\text{I}},
n_{\text{D}}\}.\end{align*}

Case 2: $n_{\text{I}}>n_{\text{D}}$. In this case, there exists an
injection
$\chi':\textbf{X}_{\text{D}}\rightarrow\textbf{Y}_{\text{I}}$ and
we can let
$\bar{\chi}:\textbf{X}\backslash\textbf{Y}\rightarrow\textbf{Y}
\backslash\textbf{X}$ be such that
\begin{equation*}
\bar{\chi}(\textbf{x})=\left\{\begin{aligned}
&\chi(\textbf{x})& ~ ~\text{if}~\textbf{x}\in\textbf{X}_{\text{S}}; \\
&\chi'(\textbf{x})& ~ ~ _{~}\text{if}~\textbf{x}\in\textbf{X}_{\text{D}}.\\
\end{aligned} \right. \label{eqn:8}
\end{equation*}
Since $|\textbf{Y}
\backslash\textbf{X}|-|\textbf{X}\backslash\textbf{Y}|=|\textbf{Y}_{\text{I}}|
-|\textbf{X}_{\text{D}}|=n_{\text{I}}-n_{\text{D}}$, then by
\eqref{eq-def-chi-d},
\begin{align*}
d_{\bar{\chi}}(\textbf{X}\backslash\textbf{Y},\textbf{Y}\backslash\textbf{X})
&=\sum_{\textbf{x}\in\textbf{X}\backslash\textbf{Y}}
d_{\text{H}}(\textbf{x},\bar{\chi}(\textbf{x}))+L\cdot\left(|\textbf{Y}
\backslash\textbf{X}|-|\textbf{X}\backslash\textbf{Y}|\right)\\
&=\sum_{\textbf{x}\in\textbf{X}_{\text{S}}}
d_{\text{H}}(\textbf{x},\chi(\textbf{x}))
+\sum_{\textbf{x}\in\textbf{X}_{\text{D}}}
d_{\text{H}}(\textbf{x},\chi'(\textbf{x}))\\
&~~~~+L\cdot\left(n_{\text{I}}-n_{\text{D}}\right)\\
&\leq n_{\text{S}}+L\cdot
n_{\text{D}}+L\cdot (n_{\text{I}}-n_{\text{D}})\\
&= n_{\text{S}}+L\cdot n_{\text{I}}\\
&=n_{\text{S}}+L\cdot\max\{n_{\text{I}},
n_{\text{D}}\}.\end{align*} Hence, similar to Case 1, we have
\begin{align*}
d_{\text{S}}(\textbf{X},\textbf{Y})
&=d_{\text{S}}(\textbf{Y}\backslash\textbf{X},
\textbf{X}\backslash\textbf{Y})\\&\leq
d_{\bar{\chi}}(\textbf{X}\backslash\textbf{Y},
\textbf{Y}\backslash\textbf{X})\\
&\leq n_{\text{S}}+L\cdot\max\{n_{\text{I}},
n_{\text{D}}\}.\end{align*}

In both cases, we have $d_{\text{S}}(\textbf{X},\textbf{Y})\leq
n_{\text{S}}+L\cdot\max\{n_{\text{I}}, n_{\text{D}}\}$, which
completes the proof.
\end{proof}

Equality in the bound of $d_{\text{S}}(\textbf{X},\textbf{Y})$ in
Lemma \ref{lem-ep-dst} can be achieved. As an example, consider
$\mathbb A=\{0,1\}$ and $L=4$, and let the input
$\textbf{X}=\{0011,1010\}$ and output
$\textbf{Y}=\{0111,1010,1100\}$, where $0111$ is an erroneous copy
of $0011$ with one substitution and $1100$ is an inserted
sequence. Then the error pattern of $\textbf{Y}$ is
$(n_{\text{I}},n_{\text{D}},n_{\text{S}})=(1,0,1)$ and so
$n_{\text{S}}+L\cdot\max\{n_{\text{I}}, n_{\text{D}}\}$. On the
other hand, by \eqref{eq-def-dst}, we can easily obtain
$d_{\text{S}}(\textbf{X},\textbf{Y})=5$. Hence, we have
$d_{\text{S}}(\textbf{X},\textbf{Y})=n_{\text{S}}+L\cdot\max\{n_{\text{I}},
n_{\text{D}}\}$.

For the decoder, when receiving a subset
$\textbf{\text{Y}}\subseteq\mathbb A^L$, its task is to find a
possible input subset $\hat{\textbf{\text{X}}}\subseteq\mathbb
A^L$ that is most similar to $\textbf{\text{Y}}$. By the above
discussion and Corollary \ref{cor-dst}, clearly, the
sequence-subset distance is a good choice of metric for similarity
between $\textbf{\text{Y}}$ and $\hat{\textbf{\text{X}}}$. In the
next subsection, we will discuss error-correcting in DNA storage
channel using codes with respect to sequence-subset distance.

\subsection{Codes with Sequence-Subset Distance}
A \emph{sequence-subset code} over $\mathbb A^L$ is a subset
$\mathcal C$ of the power set $\mathcal P(\mathbb A^L)$ of the set
$\mathbb A^L$. 
We call each element of $\mathbb A^L$ a \emph{sequence} and call
$L$ the \emph{sequence length} of $\mathcal C$. The size
$|\mathcal C|$ of $\mathcal C$ is called the \emph{code size} of
$\mathcal C$. In contrast, for each codeword
$\textbf{\text{X}}\in\mathcal C$, the size of
$\textbf{\text{X}}~($i.e., the number of sequences contained in
$\textbf{\text{X}})$ is called the \emph{codeword size} of
$\mathcal C$.

Note that a sequence-subset code $\mathcal C$ may contain
codewords of different sizes. The maximum of codeword sizes of
$\mathcal C$, i.e., $M=\max\{|\textbf{\text{X}}|;
\textbf{\text{X}}\in\mathcal C\}$, is called the \emph{maximal
codeword size} of $\mathcal C$. A sequence-subset code $\mathcal
C$ is said to have \emph{constant codeword size} (a
\emph{constant-codeword-size code}) if all codewords of $\mathcal
C$ have the same codeword size. Real DNA storage systems usually
use codes with constant codeword size. In fact, using
constant-codeword-size codes, the decoder knows how many sequences
are stored and hence can easily determine how many sequences are
lost during the synthesizing/sequencing process. However, in this
work, for the generality of the theory, we allow that different
codewords may have different sizes.

The rate of $\mathcal C$ is defined as
$$R(\mathcal C)=\frac{\log_q|\mathcal C|}
{\log_q\left(\sum_{m=0}^M{q^L\choose m}\right)}$$ and the
redundancy of $\mathcal C$ is defined as
$$r(\mathcal C)=\log_q\left(\sum_{m=0}^M{q^L\choose
m}\right)-\log_q|\mathcal C|,$$ where $q=|\mathbb A|$ and
$\sum_{m=0}^M{q^L\choose m}$ is the number of all subsets of
$\mathbb A^L$ of size not greater than $M$.\footnote{In
\cite{Heckel}, the storage rate of a code $\mathcal C$ is defined
as $\frac{\log|\mathcal C|}{ML}$, where $M$ is the constant
codeword size. The definition of the rate $R(\mathcal C)$ in this
work is slightly different from the traditional definition of code
rate. This is because for general sequence-subset codes, the size
of different codewords may be different. Hence we use
$\log_q\left(\sum_{m=0}^M{q^L\choose m}\right)$ rather than $ML$
in the definition of $R(\mathcal C)$, where $M$ is the maximal
codeword size of $\mathcal C$.} If $\mathcal C\subseteq\mathbb
A^L$ is a code with constant codeword size $M$, then the rate and
redundancy of $\mathcal C$ are defined as
$$R(\mathcal C)=\frac{\log_q|\mathcal C|}
{\log_q{q^L\choose M}}$$ and
$$r(\mathcal C)=\log_q{q^L\choose M}-\log_q|\mathcal C|,$$
respectively, where ${q^L\choose M}$ is the number of all subsets
of $\mathbb A^L$ of size $M$.

The \emph{minimum distance} of a sequence-subset code $\mathcal
C$, denoted by $d_{\text{S}}(\mathcal C)$, is the minimum of the
sequence-subset distance between any two distinct codewords of
$\mathcal C$, that is,
$$d_{\text{S}}(\mathcal C)=\min\{d_{\text{S}}
(\textbf{\text{X}},\textbf{\text{X}}');
\textbf{\text{X}},\textbf{\text{X}}'\in\mathcal C \text{~and~}
\textbf{\text{X}}\neq\textbf{\text{X}}'\}.$$

In general, $L,M,|\mathcal C|$ and $d_{\text{S}}(\mathcal C)$ are
four main parameters of $\mathcal C$, and we will call $\mathcal
C$ an $\left(L,M,|\mathcal C|, d_{\text{S}}(\mathcal C)\right)_q$
code, where $q$ is the size of the alphabet $\mathbb A$.

\vspace{0.1cm}Suppose $\mathcal C\subseteq\mathcal P(\mathbb A^L)$
is a sequence-subset code. We denote $\overline{\mathcal
C}=\{\overline{\textbf{\text{X}}}; \textbf{\text{X}}\in\mathcal
C\}$, where $\overline{\textbf{\text{X}}}=\mathbb
A^L\backslash\textbf{\text{X}}$. By Corollary \ref{cor-dst}, for
any $\textbf{\text{X}}_1,\textbf{\text{X}}_2\in\mathcal C$, we
have $d_{\text{S}}(\textbf{\text{X}}_1,\textbf{\text{X}}_2)=
d_{\text{S}}(\textbf{\text{X}}_1\backslash\textbf{\text{X}}_2,
\textbf{\text{X}}_2\backslash\textbf{\text{X}}_1)=
d_{\text{S}}(\overline{\textbf{\text{X}}_1},\overline{\textbf{\text{X}}_2})$.
Hence, $\overline{\mathcal C}$ and $\mathcal C$ have the same
sequence length $L$, code size $|\overline{\mathcal C}|=|\mathcal
C|$ and minimum distance $d_{\text{S}}(\mathcal
C)=d_{\text{S}}(\overline{\mathcal C})$. For sequence-subset code
with constant codeword size $M$, it is assumed that
$M\leq\frac{|\mathbb A|^L}{2}$. Otherwise, we can consider
$\overline{\mathcal C}$, which has constant codeword size
$\overline{M}=|\mathbb A|^L-M\leq\frac{|\mathbb A|^L}{2}$.

A\emph{minimum-distance decoder} for $\mathcal C$ is a function
$D:\mathcal P(\mathbb A^L)\rightarrow\mathcal C$ such that for any
$\textbf{\text{Y}}\in\mathcal P(\mathbb A^L)$,
\begin{align*}D(\textbf{\text{Y}})=\arg\min_{\textbf{X}'\in\mathcal
C}d_{\text{S}}(\textbf{X}',\textbf{Y}).\end{align*}

{\vskip 3pt}\begin{thm}\label{thm-E-C-Ability} Suppose $\mathcal
C$ has minimum distance $d_{\text{S}}(\mathcal C)$ and
\begin{align}\label{eq1-E-C-Ability}
n_{\text{S}}+L\cdot\max\{n_{\text{I}},n_{\text{D}}\}\leq
\frac{d_{\text{S}}(\mathcal C)-1}{2}.\end{align} Then any error of
pattern $(n_{\text{I}},n_{\text{D}},n_{\text{S}})$ can be
corrected by the minimum-distance decoder for $\mathcal C$.
\end{thm}
\begin{proof}
Let $\textbf{\text{X}}$ be the set of input sequences and
$\textbf{\text{Y}}$ be the set of output sequences of the DNA
storage channel. By Lemma \ref{lem-ep-dst}, if $\textbf{\text{Y}}$
has error pattern $(n_{\text{I}},n_{\text{D}},n_{\text{S}})$, then
$$d_{\text{S}}(\textbf{\text{X}},\textbf{\text{Y}})\leq
n_{\text{S}}+L\cdot\max\{n_{\text{I}},n_{\text{D}}\}.$$ Combining
this with \eqref{eq1-E-C-Ability}, we have
$$d_{\text{S}}(\textbf{\text{X}},\textbf{\text{Y}})\leq
\frac{d_{\text{S}}(\mathcal C)-1}{2}.$$ Hence,
$\textbf{\text{X}}=\arg\min_{\textbf{X}'\in\mathcal
C}d_{\text{S}}(\textbf{X}',\textbf{\text{Y}})=D(\textbf{\text{Y}})$,
that is, $\textbf{\text{X}}$ can be correctly recovered by the
minimum-distance decoder.
\end{proof}

\begin{exam}\label{exam-model}
Consider $\mathbb A=\{0,1\}$ and $L=5$. Let $\mathcal
C=\{\textbf{X}_1,\textbf{X}_2,\textbf{X}_3\}$, where
$\textbf{X}_1=\{00101,10001\}$, $\textbf{X}_2=\{01011,10110\}$ and
$\textbf{X}_3=\{01000,11100\}$. According to \eqref{eq-def-dst},
we can obtain $d_{\text{S}}(\textbf{X}_1,\textbf{X}_2)=
d_{\text{S}}(\textbf{X}_1,\textbf{X}_3)=6$ and
$d_{\text{S}}(\textbf{X}_2,\textbf{X}_3)=4$, hence we have
$d_{\text{S}}(\mathcal C)=4$. By Theorem \ref{thm-E-C-Ability}, if
$\textbf{X}\in\mathcal C$ is stored $($i.e., $\textbf{X}$ is the
input of the DNA storage channel$)$, $\textbf{Y}\subseteq\mathbb
A^5$ is the read result $($i.e., the output of the channel$)$ and
the error pattern $(n_{\text{I}},n_{\text{D}},n_{\text{S}})$ of
$\textbf{Y}$ satisfies
$n_{\text{S}}+L\cdot\max\{n_{\text{I}},n_{\text{D}}\}\leq
\frac{d_{\text{S}}(\mathcal C)-1}{2}=1$, then the minimum-distance
decoder will recover $\textbf{X}$ correctly from $\textbf{Y}$. For
example, suppose $\textbf{X}_1$ is stored and the read result is
$\textbf{Y}=\{01101,10001\}$, where $00101$ is changed to $01101$
by one substitution. We have
$d_{\text{S}}(\textbf{X}_1,\textbf{Y})=1$ and
$d_{\text{S}}(\textbf{X}_2,\textbf{Y})=
d_{\text{S}}(\textbf{X}_3,\textbf{Y})=5$, and then by the
minimum-distance decoder,
$D(\textbf{\text{Y}})=\arg\min_{\textbf{X}'\in\mathcal
C}d_{\text{S}}(\textbf{X}',\textbf{Y})=\textbf{X}_1$.
\end{exam}

Similar to the classical error-correcting codes, the inequality
\eqref{eq1-E-C-Ability} is a sufficient condition for an output
with error pattern $(n_{\text{I}},n_{\text{D}},n_{\text{S}})$ to
be corrected by the minimum-distance decoder but not a necessary
condition. To illustrate this, reconsider Example
\ref{exam-model}. Now, suppose $\textbf{X}_1$ is stored and the
read result is $\textbf{Y}=\{01101\}$, where $10001$ is lost and
$00101$ is changed to $01101$ by one substitution. In this case,
we have $(n_{\text{I}},n_{\text{D}},n_{\text{S}})=(0,1,1)$, and so
$n_{\text{S}}+L\cdot\max\{n_{\text{I}},n_{\text{D}}\}=6>\frac{3}{2}=
\frac{d_{\text{S}}(\mathcal C)-1}{2}$. However, $\textbf{X}_1$ can
be correctly recovered from $Y$ by the minimum-distance decoder
because $d_{\text{S}}(\textbf{X}_1,\textbf{Y})=6<7
=d_{\text{S}}(\textbf{X}_2,\textbf{Y})=
d_{\text{S}}(\textbf{X}_3,\textbf{Y})$.

\begin{rem}
Suppose the channel input is $\textbf{X}$ and output is
$\textbf{Y}$ such that the error pattern of $\textbf{Y}$ is
$(n_{\text{I}}, n_{\text{D}}, n_{\text{S}})$. It is sufficient to
assume that $n_{\text{I}}=0$ or $n_{\text{D}}=0$. In fact, suppose
$n_{\text{I}}\geq n_{\text{D}}$. Denoted by
$\textbf{x}_1,\cdots,\textbf{x}_{n_{\text{D}}}$ the $n_{\text{D}}$
deleted sequences and
$\textbf{y}_1,\cdots,\textbf{y}_{n_{\text{I}}}$ the $n_{\text{I}}$
inserted sequences. For each $i\in[n_{\text{D}}]$, we view
$\textbf{y}_i$ as obtained from $\textbf{x}_i$ by symbol
substitutions. Then the error pattern is
$(n_{\text{I}}',n_{\text{D}}',n_{\text{S}}')
=(n_{\text{I}}-n_{\text{D}},0,n_{\text{S}}+\overline{n_{\text{S}}})$,
where $\overline{n_{\text{S}}}$ is the total substitutions in
$\textbf{y}_1,\cdots,\textbf{y}_{n_{\text{D}}}$. Noticing that the
number of substitutions in each sequence is not greater than its
length $L$, then the total substitutions
$\overline{n_{\text{S}}}\leq n_{\text{D}}L$, and hence
$n_{\text{S}}'+L\max\{n_{\text{I}}',n_{\text{D}}'\}
=n_{\text{S}}+\overline{n_{\text{S}}}+(n_{\text{I}}-n_{\text{D}})L\leq
n_{\text{S}}+n_{\text{D}}L+(n_{\text{I}}-n_{\text{D}})L=
n_{\text{S}}+n_{\text{I}}L=n_{\text{S}}+
L\max\{n_{\text{I}},n_{\text{D}}\}$. By Theorem
\ref{thm-E-C-Ability}, if $\textbf{Y}$ with error pattern
$(n_{\text{I}},n_{\text{D}},n_{\text{S}})$ can be corrected, then
$\textbf{Y}$ with error pattern
$(n_{\text{I}}',n_{\text{D}}',n_{\text{S}}')$ can also be
corrected. Similarly, if $n_{\text{I}}<n_{\text{D}}$, then by
viewing each $\textbf{y}_i$, $i\in[n_{\text{I}}]$, as obtained
from $\textbf{x}_i$ by symbol substitutions, we can obtain an
error pattern $(n_{\text{I}}',n_{\text{D}}',n_{\text{S}}')
=(0,n_{\text{D}}-n_{\text{I}},n_{\text{S}}+\overline{n_{\text{S}}})$
such that if $\textbf{Y}$ with error pattern
$(n_{\text{I}},n_{\text{D}},n_{\text{S}})$ can be corrected, then
$\textbf{Y}$ with error pattern
$(n_{\text{I}}',n_{\text{D}}',n_{\text{S}}')$ can also be
corrected. Thus, we can always assume that $n_{\text{I}}=0$ or
$n_{\text{D}}=0$.
\end{rem}

In \cite{Lenz18} and \cite{Sima18}, it was assumed that the number
of output sequences is not greater than the number of input
sequences. In this work, considering the sequence insertion
errors, we allow that the number of output sequences of the DNA
storage channel can be larger than the number of input sequences.
This assumption is also of interest for the generality of the
theory.

Usually, correction of sequence insertion/deletion requires codes
with larger minimum distance. To see this, consider two special
cases of error pattern. The first special case is
$n_{\text{I}}=n_{\text{D}}=0$. In this case, by Theorem 2, a
sequence-subset code $\mathcal C$ with minimum distance
$d_{\text{S}}(\mathcal C)$ can correct a total of at most
$\frac{d_{\text{S}}(\mathcal C)-1}{2}$ substitution errors. The
second special case is $n_{\text{I}}=n_{\text{S}}=0$. A
sequence-subset code $\mathcal C$ with minimum distance
$d_{\text{S}}(\mathcal C)$ can correct a total of at most
$\frac{d_{\text{S}}(\mathcal C)-1}{2L}$ sequence-deletions, which
is $\frac{1}{L}$ of the number of correctable substitution errors.

\section{Bounds on the Size of Sequence Subset Codes}

In this Section, we consider codes with constant codeword size.
Let $S_q(L,M,d)$ denote the maximum number of codewords in a
sequence-subset code over a $q$-ary alphabet with sequence length
$L$, constant codeword size $M$ and minimum sequence-subset
distance at least $d$. A $q$-ary sequence-subset code is said to
be optimal (with respect to code size) if it has the largest
possible code size of any $q$-ary sequence-subset code of the
given parameters $L,M$ and $d$. In this section, we always assume
that $\mathbb A$ is an alphabet of size $q$. We will derive some
upper bounds on $S_q(L,M,d)$.

Clearly, for any sequence-subset code $\mathcal C\subseteq\mathcal
P(\mathbb A^{L})$ with constant codeword size $M$, its minimum
distance $d_{\text{S}}(\mathcal C)\leq LM$, and hence
$M\geq\frac{d_{\text{S}}(\mathcal C)}{L}$.
For this reason, in the following, we always assume that $d\leq
LM$, or equivalently, $M\geq\left\lceil\frac{d}{L}\right\rceil$.

\subsection{Upper Bound for the Special Case $d=LM$}

First, consider the special case that \textcolor{blue}{$d=LM$}. We
have the following upper bound on $S_q(L,M,d).$

\begin{thm}\label{thm-bnd-S-M0}
Suppose $d=LM$. Then
\begin{align}\label{eq-bnd-S-M0}S_q(L,M,d)\leq
\left\lfloor qM^{-\frac{1}{L}}\right\rfloor.\end{align}
\end{thm}
\begin{proof}
Let $\mathcal C=\{\textbf{\text{X}}_1, \textbf{\text{X}}_2,
\cdots, \textbf{\text{X}}_N\}\subseteq\mathcal P(\mathbb A^{L})$
be an arbitrary sequence-subset code with constant codeword size
$M$ and minimum distance $d$, where for each $i\in[N]$,
$\textbf{\text{X}}_i=\{\textbf{\text{x}}_{i,1},
\textbf{\text{x}}_{i,2}, \cdots,
\textbf{\text{x}}_{i,M}\}\subseteq\mathbb A^{L}$. We will prove
that $N\leq qM^{-\frac{1}{L}}$.

For each fixed $\ell\in[L]$ and $i\in[N]$, let
$$W_{i,\ell}=\bigcup_{j\in[M]}\{\textbf{\text{x}}_{i,j}(\ell)\}.$$
Then we have the following Claim.

\emph{Claim 1}: For each fixed $\ell\in[L]$ and $i\in[N]$,
$W_{1,\ell}, W_{2,\ell}, \cdots, W_{N,\ell}$ are mutually disjoint
subsets of $\mathbb A$.

To prove Claim 1, we first notice that for any distinct
$i_1,i_2\in[N]$ and any (not necessarily distinct)
$j_1,j_2\in[M]$,
\begin{align}\label{eq0-vlu-S-M0}
d_{\text{H}}(\textbf{x}_{i_1,j_1},
\textbf{x}_{i_2,j_2})=L,\end{align} which can be proved as
follows. Since both $\textbf{x}_{i_1,j_1}$ and
$\textbf{x}_{i_2,j_2}$ have length $L$, we have
$d_{\text{H}}(\textbf{x}_{i_1,j_1}, \textbf{x}_{i_2,j_2})\leq L$.
We can only have $d_{\text{H}}(\textbf{x}_{i_1,j_1},
\textbf{x}_{i_2,j_2})=L$ because otherwise, we have
$d_{\text{H}}(\textbf{\text{x}}_{i_1,j_1},
\textbf{\text{x}}_{i_2,j_2})<L$ and we can construct a bijection
$\chi: \textbf{X}_{i_1}\rightarrow\textbf{X}_{i_2}$ such that
$\chi(\textbf{x}_{i_1,j_1})=\textbf{x}_{i_2,j_2}$. Since for all
$\textbf{x}_{i_1,j}\in\textbf{X}_{i_1}$ and
$\textbf{x}_{i_1,j'}\in\textbf{X}_{i_2}$,
$d_{\text{H}}(\textbf{x}_{i_1,j}, \textbf{x}_{i_1,j'})\leq L$ (the
length of $\textbf{x}_{i_1,j}$ and $\textbf{x}_{i_2,j'}$), then by
\eqref{eq-def-chi-d} we can obtain $d_{\chi}(\textbf{X}_{i_1},
\textbf{X}_{i_2})<LM$, and further by Definition \ref{def-dst} we
have $d_{\text{S}}(\textbf{X}_{i_1}, \textbf{X}_{i_2})<LM$, which
contradicts to the assumption that the minimum distance of
$\mathcal C$ is $d=LM$. Hence, \eqref{eq0-vlu-S-M0} must hold.
Again since both $\textbf{x}_{i_1,j_1}$ and $\textbf{x}_{i_2,j_2}$
have length $L$, then \eqref{eq0-vlu-S-M0} implies that for any
fixed $\ell\in[L]$,
$\textbf{\text{x}}_{i_1,j_1}(\ell)\neq\textbf{\text{x}}_{i_2,j_2}(\ell)$.
Since $j_1,j_2$ are any elements of $[M]$, then we have
$W_{i_1,\ell}\cap W_{i_2,\ell}=\emptyset$. Further, since
$i_1,i_2$ are any distinct elements of $[N]$, we have $W_{1,\ell},
W_{2,\ell}, \cdots, W_{N,\ell}$ are mutually disjoint subsets of
$\mathbb A$, which proves Claim 1.

Now, by Claim 1, we have
\begin{align}\label{eq1-vlu-S-M0}\sum_{i=1}^N|W_{i,\ell}|\leq|\mathbb
A|=q.\end{align}

By the construction of $W_{i,\ell}$, for each $i\in[N]$ and
$j\in[M]$, we have $\textbf{\text{x}}_{i,j}\in W_{i,1}\times
W_{i,2}\times\cdots\times W_{i,L}$, which implies that
$\textbf{\text{X}}_i=\{\textbf{\text{x}}_{i,1},
\textbf{\text{x}}_{i,2}, \cdots,
\textbf{\text{x}}_{i,M}\}\subseteq W_{i,1}\times
W_{i,2}\times\cdots\times W_{i,L}$, and hence we have
\begin{align}\label{eq2-vlu-S-M0}|W_{i,1}\times W_{i,2}\times\cdots\times
W_{i,L}|=\prod_{\ell=1}^L|W_{i,\ell}|\geq
|\textbf{\text{X}}_i|=M.\end{align}

Now, consider \eqref{eq1-vlu-S-M0}. By the inequality of
arithmetic and geometric means, for each $\ell\in[L]$, we have
\begin{align*}
\frac{q}{N}\geq\frac{1}{N}\sum_{i=1}^N|W_{i,\ell}|
\geq\left(\prod_{i=1}^N|W_{i,\ell}|\right)^{\frac{1}{N}}.\end{align*}
Combining this with \eqref{eq2-vlu-S-M0}, we have
\begin{align*}
\left(\frac{q}{N}\right)^L&
\geq\prod_{\ell=1}^L\left(\prod_{i=1}^N|W_{i,\ell}|\right)^{\frac{1}{N}}
\nonumber\\
&=\prod_{i=1}^N\left(\prod_{\ell=1}^L|W_{i,\ell}|\right)^{\frac{1}{N}}
\nonumber\\
&\geq(M^{\frac{1}{N}})^N\nonumber\\&=M.\end{align*} From this we
have $\frac{q}{N}\geq M^{\frac{1}{L}}$, which implies $N\leq
qM^{-\frac{1}{L}}$. Hence, $$S_q(L,M,d)\leq qM^{-\frac{1}{L}}.$$
Since $S_q(L,M,d)$ is an integer, so
$$S_q(L,M,d)\leq\left\lfloor qM^{-\frac{1}{L}}\right\rfloor,$$
which completes the proof.
\end{proof}

Consider the redundancy of the codes. For any code $\mathcal
C\subseteq\mathbb A^L$ with constant codeword size $M$ and minimum
distance $d=LM$, by Theorem \ref{thm-bnd-S-M0}, we have
\begin{align*} r(\mathcal C)&=\log_q{q^L\choose M}
-\log_q|\mathcal C|\\&\geq \log_q{q^L\choose M}-\log_q(qM^{-\frac{1}{L}})\\
&=\log_q\frac{q^L!M^{\frac{1}{L}}}{M!(q^L-M)!q}.\end{align*}

The bound \eqref{eq-bnd-S-M0} is tight for the special case that
$M^{\frac{1}{L}}$ is an integer. Codes that achieve equality of
\eqref{eq-bnd-S-M0} are constructed in Section
\uppercase\expandafter{\romannumeral 4}.A.

\subsection{Plotkin-like Bound}
We present the Plotkin-like Bound of sequence-subset codes as the
following theorem.
\begin{thm}[Plotkin-like Bound]\label{thm-pltk-bnd}
Let $\mathcal C$ be an $\left(L,M,N,d\right)_q$ code such that
$rLM<d$, where $r=1-\frac{1}{q}$. Then
$$N\leq\frac{d}{d-rLM}.$$
\end{thm}
\begin{proof}
Our proof of this theorem is similar to the proof of \cite[Theorem
2.2.1]{Huffman}.

Suppose $\mathcal C=\{\textbf{\text{X}}_1, \textbf{\text{X}}_2,
\cdots, \textbf{\text{X}}_N\}$ such that for each $i\in[N]$,
$\textbf{\text{X}}_i=\{\textbf{\text{x}}_{i,1},
\textbf{\text{x}}_{i,2}, \cdots,
\textbf{\text{x}}_{i,M}\}\subseteq\mathbb A^{L}$. First, we have
the following claim, which we will prove later.\vspace{0.1cm}

\emph{Claim 2}: For any distinct $i_1,i_2\in[N]$, we have
\begin{align*}
d_{\text{S}}(\textbf{\text{X}}_{i_1},
\textbf{\text{X}}_{i_2})\leq\frac{1}{M}\sum_{j_1,j_2\in[M]}
d_{\text{H}}(\textbf{\text{x}}_{i_1,j_1},
\textbf{\text{x}}_{i_2,j_2}).\end{align*}

Now, let $$A=\sum_{i_1,i_2\in[N]}\sum_{j_1,j_2\in[M]}
d_{\text{H}}(\textbf{\text{x}}_{i_1,j_1},
\textbf{\text{x}}_{i_2,j_2}).$$ Since $d$ is the minimum distance
of $\mathcal C$, by the averaging principle \cite{Jukna}, we have
\begin{align}\label{eq2-sum-dsc}
d&\leq{N\choose
2}^{-1}\sum_{\{i_1,i_2\}\subseteq[N]}d_{\text{S}}(\textbf{\text{X}}_{i_1},
\textbf{\text{X}}_{i_2})\nonumber\\&=\frac{1}{2}{N\choose
2}^{-1}\sum_{i_1,i_2\in[N], i_1\neq
i_2}d_{\text{S}}(\textbf{\text{X}}_{i_1},
\textbf{\text{X}}_{i_2})\nonumber\\&\leq\frac{1}{2}{N\choose
2}^{-1}\sum_{i_1,i_2\in[N]}d_{\text{S}}(\textbf{\text{X}}_{i_1},
\textbf{\text{X}}_{i_2})\nonumber\\&\leq\frac{1}{N(N-1)}
\sum_{i_1,i_2\in[N]}\left(\frac{1}{M}\sum_{j_1,j_2\in[M]}
d_{\text{H}}(\textbf{\text{x}}_{i_1,j_1},
\textbf{\text{x}}_{i_2,j_2})\right)\nonumber\\
&=\frac{1}{N(N-1)}\frac{1}{M}\cdot A,\end{align} where the last
inequality is obtained by Claim 2.

For each $a\in\mathbb A$ and $\ell\in[L]$, let $n_{\ell,a}$ be the
number of $(i,j)\in[N]\times[M]$ such that
$\textbf{\text{x}}_{i,j}(\ell)=a$. Then for each fixed
$\ell\in[L]$, we have
\begin{align}\label{eq5-sum-dsc}
\sum_{a\in\mathbb A}n_{\ell,a}=NM.\end{align} Moreover, we have
\begin{align}\label{eq3-sum-dsc}
A&=\sum_{i_1,i_2\in[N]}\sum_{j_1,j_2\in[M]}
d_{\text{H}}(\textbf{\text{x}}_{i_1,j_1},
\textbf{\text{x}}_{i_2,j_2})\nonumber\\&=\sum_{\ell=1}^L\sum_{a\in\mathbb
A}n_{\ell,a}(NM-n_{\ell,a})\nonumber\\
&=L(NM)^2-\sum_{\ell=1}^L\sum_{a\in\mathbb
A}n_{\ell,a}^2.\end{align} For each $\ell\in[L]$, by the
Cauchy-Schwartz inequality,
\begin{align*}
\left(\sum_{a\in\mathbb A}n_{\ell,a}\right)^2\leq
q\sum_{a\in\mathbb A}n_{\ell,a}^2,\end{align*} where $q=|\mathbb
A|$. Combining this with \eqref{eq3-sum-dsc}, we obtain
\begin{align}\label{eq4-sum-dsc}
A&\leq L(NM)^2-\sum_{\ell=1}^L\frac{1}{q}\left(\sum_{a\in\mathbb
A}n_{\ell,a}\right)^2\nonumber\\&=L(NM)^2-\sum_{\ell=1}^L\frac{1}{q}
\left(NM\right)^2\nonumber\\&=\left(1-\frac{1}{q}\right)L(NM)^2,\end{align}
where the first equality is obtained from \eqref{eq5-sum-dsc}.
Combining \eqref{eq2-sum-dsc} and \eqref{eq4-sum-dsc}, we obtain
$$d\leq\frac{1}{N(N-1)}\frac{1}{M}\cdot
\left(1-\frac{1}{q}\right)L(NM)^2.$$ Solving $N$ from the above
inequality we obtain
$$N\leq\frac{d}{d-rLM}~,$$ where $r=1-\frac{1}{q}$.
\end{proof}

\vspace{0.1cm} To complete the proof of Theorem
\ref{thm-pltk-bnd}, we still need to prove Claim 2.
\begin{proof}[Proof of Claim 2]
Let $\mathscr S_M$ denote the permutation group on $[M]$. Note
that for any $j_1,j_2\in[M]$, not necessarily distinct, there are
$(M-1)!$ permutations $\chi\in\mathscr S_M$ such that
$\chi(j_1)=j_2$. We have
\begin{align}\label{eq1-claim1}&\sum_{\chi\in\mathscr
S_M}\sum_{j\in[M]}d_{\text{H}}(\textbf{\text{x}}_{i_1,j},
\textbf{\text{x}}_{i_2,\chi(j)})\nonumber\\&=(M-1)!\sum_{j_1,j_2\in[M]}
d_{\text{H}}(\textbf{\text{x}}_{i_1,j_1},
\textbf{\text{x}}_{i_2,j_2}).\end{align} Further, by Definition
\ref{def-dst} and the averaging principle \cite{Jukna}, we have
\begin{align*}
d_{\text{S}}(\textbf{\text{X}}_{i_1},
\textbf{\text{X}}_{i_2})&\leq\frac{1}{M!}\sum_{\chi\in\mathscr
S_M}d_{\chi}(\textbf{\text{X}}_{i_1},
\textbf{\text{X}}_{i_2})\nonumber\\&=\frac{1}{M!}\sum_{\chi\in\mathscr
S_M}\sum_{j\in[M]}d_{\text{H}}(\textbf{\text{x}}_{i_1,j},
\textbf{\text{x}}_{i_2,\chi(j)})\nonumber\\
&=\frac{(M-1)!}{M!}\sum_{j_1,j_2\in[M]}
d_{\text{H}}(\textbf{\text{x}}_{i_1,j_1},
\textbf{\text{x}}_{i_2,j_2})\nonumber\\&=\frac{1}{M}\sum_{j_1,j_2\in[M]}
d_{\text{H}}(\textbf{\text{x}}_{i_1,j_1},
\textbf{\text{x}}_{i_2,j_2}),\end{align*} where the second
equality comes from \eqref{eq1-claim1}.
\end{proof}

Consider the redundancy of the codes. For any code $\mathcal
C\subseteq\mathbb A^L$ with constant codeword size $M$ and minimum
distance $d>rLM$, where $r=1-\frac{1}{q}$, by Theorem
\ref{thm-pltk-bnd}, we have \begin{align*} r(\mathcal
C)&=\log_q{q^L\choose M}-\log_q|\mathcal C|\\&\geq
\log_q{q^L\choose M}-\log_q(\frac{d}{d-rLM})\\
&=\log_q\frac{q^L!(d-rLM)}{M!(q^L-M)!d}.\end{align*}

For the special case that $d=LM$, we have
$\frac{d}{d-rLM}=q>qM^{-\frac{1}{L}}$, where $r=1-\frac{1}{q}$.
Thus, the bound given in Theorem \ref{thm-bnd-S-M0} is tighter
than the bound given in Theorem \ref{thm-pltk-bnd} for $M>1$. It
is still not known whether the bound in Theorem \ref{thm-pltk-bnd}
is achievable for other cases.

\subsection{Singleton-like Bound}

For each code $\mathcal C=\{\textbf{\text{X}}_1,
\textbf{\text{X}}_2, \cdots,
\textbf{\text{X}}_N\}\subseteq\mathcal P(\mathbb A^{L})$, denote
\begin{align}\label{cup-X-C}V(\mathcal
C)=\bigcup_{i=1}^N\textbf{\text{X}}_i.\end{align} Further, let
$\tilde{S}_q(L,M,K,d)$ denote the maximum number of codewords in a
sequence-subset code $\mathcal C$ over a $q$-ary alphabet $\mathbb
A$ with sequence length $L$, constant codeword size $M$, minimum
sequence-subset distance at least $d$ and $|V(\mathcal C)|\leq K$.
Clearly, for any $K\leq q^L$,
\begin{align}\label{BS-to-S}\tilde{S}_q(L,M,K,d)\leq
\tilde{S}_q(L,M,q^L,d)=S_q(L,M,d).\end{align} We first prove a
recursive bound on $\tilde{S}_q(L,M,K,d)$ \textcolor{blue}{in} the
following theorem.

\begin{thm}\label{thm-recu-M-K}
Suppose $d\leq LM$ and $K\leq q^L$. We have
\begin{align}\label{eq1-recu-M-K}\tilde{S}_q(L,M,K,d)\leq
\left\lfloor\frac{K}{M}\tilde{S}_q(L,M-1,K-1,d)\right\rfloor.\end{align}
\end{thm}
\begin{proof}
Let $\mathcal C=\{\textbf{\text{X}}_1, \textbf{\text{X}}_2,
\cdots, \textbf{\text{X}}_N\}\subseteq\mathcal P(\mathbb A^{L})$
be a sequence-subset code with constant codeword size $M$, minimum
distance at least $d$ such that $|V(\mathcal C)|\leq K$ and code
size $|\mathcal C|=N=\tilde{S}_q(L,M,K,d)$, where
$\textbf{\text{X}}_i\subseteq\mathbb A^{L}$ for each $i\in[N]$.

For each $\textbf{\text{x}}\in V(\mathcal C)$, let
$$\mathcal C(\textbf{\text{x}})=\{\textbf{\text{X}}\in\mathcal C;
\textbf{\text{x}}\in\textbf{\text{X}}\}$$ and $$\tilde{\mathcal
C}(\textbf{\text{x}})=\{\tilde{\textbf{\text{X}}}
=\textbf{\text{X}}\backslash\{\textbf{\text{x}}\};
\textbf{\text{X}}\in\mathcal C(\textbf{\text{x}})\}.$$ Then
$\tilde{\mathcal C}(\textbf{\text{x}})\subseteq\mathcal P(\mathbb
A^{L})$ has constant codeword size $M-1$, size $|\tilde{\mathcal
C}(\textbf{\text{x}})|=|\mathcal C(\textbf{\text{x}})|$ and
$|V(\tilde{\mathcal C}(\textbf{\text{x}}))|\leq K-1$.

Moreover, for any distinct $\tilde{\textbf{\text{X}}}_{i_1},
\tilde{\textbf{\text{X}}}_{i_2}\in\tilde{\mathcal
C}(\textbf{\text{x}})$, by the construction of $\tilde{\mathcal
C}(\textbf{\text{x}})$, we have $\tilde{\textbf{\text{X}}}_{i_1}
=\textbf{\text{X}}_{i_1}\backslash\{\textbf{\text{x}}\}$ and
$\tilde{\textbf{\text{X}}}_{i_2}
=\textbf{\text{X}}_{i_2}\backslash\{\textbf{\text{x}}\}$ for some
distinct
$\textbf{\text{X}}_{i_1},\textbf{\text{X}}_{i_2}\in\mathcal
C(\textbf{\text{x}})$, so
$\tilde{\textbf{\text{X}}}_{i_1}\backslash
\tilde{\textbf{\text{X}}}_{i_2}
=\textbf{\text{X}}_{i_1}\backslash\textbf{\text{X}}_{i_2}$ and
$\tilde{\textbf{\text{X}}}_{i_2}\backslash
\tilde{\textbf{\text{X}}}_{i_1}
=\textbf{\text{X}}_{i_2}\backslash\textbf{\text{X}}_{i_1}$. By
Corollary \ref{cor-dst}, we have
$$d_{\text{S}}(\tilde{\textbf{\text{X}}}_{i_1},
\tilde{\textbf{\text{X}}}_{i_2})
=d_{\text{S}}(\textbf{\text{X}}_{i_1},\textbf{\text{X}}_{i_2}).$$
Then we have $d_{\text{S}}(\tilde{\mathcal C}(\textbf{\text{x}}))=
d_{\text{S}}(\mathcal C(\textbf{\text{x}}))$. On the other hand,
since $\mathcal C(\textbf{\text{x}})\subseteq\mathcal C$, we have
$d_{\text{S}}(\mathcal C(\textbf{\text{x}}))\geq
d_{\text{S}}(\mathcal C)\geq d$. Thus,
$d_{\text{S}}(\tilde{\mathcal C}(\textbf{\text{x}}))\geq d$.

By the above discussion, for each $\textbf{\text{x}}\in V(\mathcal
C)$, we have \begin{align}\label{eq2-recu-M-K}|\tilde{\mathcal
C}(\textbf{\text{x}})|\leq\tilde{S}_q(L,M-1,K-1,d).\end{align}

Now, we estimate $|\tilde{\mathcal C}(\textbf{\text{x}})|$. Since
$|\tilde{\mathcal C}(\textbf{\text{x}})|=|\mathcal
C(\textbf{\text{x}})|$, it is sufficient to estimate $|\mathcal
C(\textbf{\text{x}})|$. Denote $V(\mathcal C)=\{\textbf{x}_1,
\textbf{x}_2, \cdots, \textbf{x}_{\tilde{K}}\}$, where
$\tilde{K}=|V(\mathcal C)|$. Consider the $N\times \tilde{K}$
matrix $I=(a_{i,j})$ such that $a_{i,j}=1$ if
$\textbf{x}_j\in\textbf{X}_i$, and $a_{i,j}=0$ otherwise. Note
that the number of ones in row $i$ of $I$ is $|\textbf{X}_i|=M$
and the number of ones in column $j$ of $I$ is $|\mathcal
C(\textbf{x}_j)|$. By counting the total number of ones in $I$, we
obtain
$$\sum_{\textbf{\text{x}}\in V(\mathcal C)}|\mathcal
C(\textbf{\text{x}})|=\sum_{\textbf{\text{X}}\in \mathcal
C}|\textbf{\text{X}}|=MN.$$ By the averaging principle
\cite{Jukna}, there exists an $\textbf{x}_{j_0}\in V(\mathcal C)$
such that
$$|\mathcal
C(\textbf{x}_{j_0})|\geq\frac{MN}{|V(\mathcal
C)|}\geq\frac{MN}{K}.$$ Hence,
\begin{align*}N\leq\frac{K}{M}|\mathcal
C(\textbf{x}_{j_0})|=\frac{K}{M}|\tilde{\mathcal
C}(\textbf{x}_{j_0})|.\end{align*} Note that $|\mathcal
C|=\tilde{S}_q(L,M,K,d)=N$. Then we have
\begin{align*}\tilde{S}_q(L,M,K,d)\leq\frac{K}{M}|\tilde{\mathcal
C}(\textbf{\text{x}}_0)|.\end{align*} This, combining with
\eqref{eq2-recu-M-K}, implies that $$\tilde{S}_q(L,M,K,d)\leq
\frac{K}{M}\tilde{S}_q(L,M-1,K-1,d).$$ Noticing that
$\tilde{S}_q(L,M,K,d)$ is an integer, then
$$\tilde{S}_q(L,M,K,d)\leq
\left\lfloor\frac{K}{M}\tilde{S}_q(L,M-1,K-1,d)\right\rfloor,$$
which completes the proof.
\end{proof}

Now, we can prove a Singleton-like bound for sequence-subset codes
as follows.
\begin{thm}[Singleton-like Bound]\label{thm-Sglnt-bnd}
Suppose $rLM_0\!<\!d\!\leq\! LM_0$, where $r=1-\frac{1}{q}$ and
$M_0=\left\lceil\frac{d}{L}\right\rceil$. Then
\begin{align*}
&S_q(L,M,d)\\
&\leq \left\lfloor\frac{q^L}{M}\left\lfloor\frac{q^L-1}{M-1}\cdots
\left\lfloor\frac{q^L-M+M_0+1}{M_0+1}f(L,M_0,d)
\right\rfloor\cdots\right\rfloor\right\rfloor,\end{align*} where
\begin{equation}
f(L,M_0,d)=\left\{\begin{aligned} &\left\lfloor
qM_0^{-\frac{1}{L}}\right\rfloor& ~ ~\text{if}~d=LM_0;
~ ~ ~ ~ ~ ~ ~ ~ ~\\
&\frac{d}{d-rLM_0}& ~ ~\text{if}~rLM_0\!<\!d\!<\! LM_0.\\
\end{aligned} \right. \label{eqn:def-f}
\end{equation}
\end{thm}
\begin{proof}
Repeatedly using Theorem \ref{thm-recu-M-K} with $M-M_0$ times, we
obtain
\begin{align*}
\tilde{S}_q(L,M,q^L,d)&\leq
\left\lfloor\frac{q^L}{M}\left\lfloor\frac{q^L-1}{M-1}\cdots
\left\lfloor\frac{q^L-M+M_0+1}{M_0+1}\right.\right.\right.\\
&~~~~\tilde{S}_q(L,M_0,q^L-M+M_0,d)
\bigg{\rfloor}\cdots\bigg{\rfloor}\bigg{\rfloor}.\end{align*}
Moreover, according to \eqref{BS-to-S}, we have
\begin{align*}S_q(L,M,d)=\tilde{S}_q(L,M,q^L,d)\end{align*} and
\begin{align*}\tilde{S}_q(L,M_0,q^L-M+M_0,d)&\leq
S_q(L,M_0,q^L,d)\\&=S_q(L,M_0,d).\end{align*} Combining the above
three equations, we have
\begin{align}S_q(L,M,d)&\leq
\left\lfloor\frac{q^L}{M}\left\lfloor\frac{q^L-1}{M-1}\cdots
\left\lfloor\frac{q^L-M+M_0+1}{M_0+1}\right.\right.\right.\nonumber\\
&~~~~S_q(L,M_0,d)\bigg{\rfloor}\cdots\bigg{\rfloor}\bigg{\rfloor}
.\label{eq2-Sglnt-bnd}\end{align}

Let $f(L,M_0,d)$ be defined as in \eqref{eqn:def-f}. By Theorem
\ref{thm-bnd-S-M0} and Theorem \ref{thm-pltk-bnd}, we have
\begin{align*}S_q(L,M_0,d)\leq
f(L,M_0,d).\end{align*} Combining this with \eqref{eq2-Sglnt-bnd},
we have \begin{align*}
S_q(L,M,d)&\leq
\left\lfloor\frac{q^L}{M}\left\lfloor\frac{q^L-1}{M-1}\cdots
\left\lfloor\frac{q^L-M+M_0+1}{M_0+1}\right.\right.\right.\\
&~~~~f(L,M_0,d) \bigg{\rfloor}\cdots\bigg{\rfloor},\end{align*}
which completes the proof.
\end{proof}

\begin{rem}\label{rem-Sglnt-bnd}
It is easy to see that
\begin{align*}&\left\lfloor\frac{q^L}{M}\left\lfloor\frac{q^L-1}{M-1}\cdots
\left\lfloor\frac{q^L-M+M_0+1}{M_0+1}f(L,M_0,d)
\right\rfloor\cdots\right\rfloor\right\rfloor\\
&\leq
\left(\prod_{k=0}^{M-M_0-1}\frac{q^L-k}{M-k}\right)f(L,M_0,d)\end{align*}
and $${q^L \choose
M}=\left(\prod_{k=0}^{M-M_0-1}\frac{q^L-k}{M-k}\right){q^L-M+M_0
\choose M_0}.$$ Hence, the bound in Theorem \ref{thm-Sglnt-bnd}
gives a bound on the code rate as
$$\frac{S_q(L,M,d)}{{q^L \choose
M}}\leq\frac{1}{{q^L-M+M_0 \choose M_0}}\cdot f(L,M_0,d),$$ where
$f(L,M_0,d)$ is defined as in \eqref{eqn:def-f}.
\end{rem}

Consider the redundancy of the codes. For any code $\mathcal
C\subseteq\mathbb A^L$ with constant codeword size $M$ and minimum
distance $d$ satisfying $rLM_0<d\leq LM_0$, where
$r=1-\frac{1}{q}$ and $M_0=\left\lceil\frac{d}{L}\right\rceil$, by
Remark \ref{rem-Sglnt-bnd}, we have
\begin{align*} S_q(L,M,d)\leq {q^L \choose M}
\frac{1}{{q^L-M+M_0 \choose M_0}}\cdot f(L,M_0,d).\end{align*}
Thus, the redundancy
\begin{align*} r(\mathcal C)&=\log_q{q^L\choose M}-\log_q|\mathcal
C|\\&\geq \log_q{q^L\choose M}-\log_q\left({q^L \choose M}
\frac{1}{{q^L-M+M_0 \choose M_0}}\cdot f(L,M_0,d)\right)\\
&=\log_q\frac{(q^L-M+M_0)!}{M_0!(q^L-M)!f(L,M_0,d)},\end{align*}
where $f(L,M_0,d)$ is defined as in \eqref{eqn:def-f}.

Clearly, Theorem \ref{thm-bnd-S-M0} is a special case of Theorem
\ref{thm-Sglnt-bnd} with $M=M_0=\lceil\frac{d}{L}\rceil$ and
$d=LM_0$. It is still not known whether the bound given in Theorem
\ref{thm-Sglnt-bnd} is achievable for the case that $d<LM_0$ or
$M>M_0$.

\section{Constructions of Sequence-Subset Codes}
In this section, we give some constructions of sequence-subset
codes. As in Section \uppercase\expandafter{\romannumeral 3}, we
will always assume that $\mathbb A$ is an alphabet of size $q$.

\subsection{Construction of Optimal Codes}
The following construction gives a family of optimal
sequence-subset code (with respect to code size) for the special
case that $d=LM$ and $M^{\frac{1}{L}}$ is an integer.

\textbf{Construction 1}: Suppose $d=LM$, $M^{\frac{1}{L}}<q$ is an
integer and $N=\left\lfloor qM^{-\frac{1}{L}}\right\rfloor$.
Partition $\mathbb A$ into $N$ mutually disjoint subsets, say
$W_1,W_2,\cdots,W_N$, such that $|W_i|\geq M^{\frac{1}{L}}$,
$i=1,2,\cdots,N$. For each $i\in[N]$, pick a subset
$\textbf{\text{X}}_i=\{\textbf{\text{x}}_{i,1},
\textbf{\text{x}}_{i,2}, \cdots,
\textbf{\text{x}}_{i,M}\}\subseteq W_i^L$, and let $\mathcal
C=\{\textbf{\text{X}}_i;i\in[N]\}$.

\begin{thm}\label{thm-extv-S-M0}
The code $\mathcal C$ obtained by Construction 1 is an
$(L,M,N,d)_q$ sequence-subset code.
\end{thm}
\begin{proof}
Since $N=\left\lfloor qM^{-\frac{1}{L}}\right\rfloor$, we have
$N\leq qM^{-\frac{1}{L}}$, and hence $NM^{\frac{1}{L}}\leq q.$ The
set $\mathbb A$ can always be partitioned into
$W_1,W_2,\cdots,W_N$ satisfying the expected conditions. Moreover,
since $|W_i^L|\geq|M|$ for all $i\in[N]$, the subsets
$\textbf{\text{X}}_i=\{\textbf{\text{x}}_{i,1},
\textbf{\text{x}}_{i,2}, \cdots,
\textbf{\text{x}}_{i,M}\}\subseteq W_i^L$, and hence $\mathcal C$,
can always be constructed as in Construction 1.

Clearly, $\mathcal C\subseteq\mathcal P(\mathbb A^{L})$ is a
sequence-subset code with constant codeword size $M$ and
$|\mathcal C|=N=\left\lfloor qM^{-\frac{1}{L}}\right\rfloor$.
Moreover, for any distinct $i_1,i_2\in[N]$ and any
$j_1,j_2\in[M]$, since $W_1,W_1,\cdots,W_N$ are mutually disjoint,
$\textbf{x}_{i_1,j_1}\in W_{i_1}^L$ and $\textbf{x}_{i_2,j_2}\in
W_{i_2}^L$, it is easy to see that
$$d_{\text{H}}(\textbf{x}_{i_1,j_1},\textbf{x}_{i_2,j_2})=L.$$
By \eqref{eq-def-chi-d} and \eqref{eq-def-dst}, for any distinct
$i_1,i_2\in[N]$, we have
$$d_{\text{S}}(\textbf{X}_{i_1},\textbf{X}_{i_2})=LM=d,$$
which implies that $d_{\text{S}}(\mathcal C)=d$. Thus, $\mathcal
C$ is an $(L,M,N,d)_q$ sequence-subset code.
\end{proof}

\begin{exam} As an illustrative example of
Construction 1, we let $\mathbb A=\{0,1,\cdots,15\}$, $L=4$ and
$M=16$. Then $q=16$, $d=64$, $M^{\frac{1}{L}}=2$ and
$N=\left\lfloor qM^{-\frac{1}{L}}\right\rfloor=8$. We partition
$\mathbb A$ into $W_1,W_2,\cdots,W_8$, where
$W_i=\{2(i-1),2i-1\}$, $i=1,2,\cdots,8$. Then we can choose
$\textbf{X}_i=W_i^4$ and let $\mathcal
C=\{\textbf{X}_1,\textbf{X}_2,\cdots,\textbf{X}_8\}$. For example,
we have $\textbf{X}_1=\{0000$, $0001$, $0010$, $0011$, $0100$,
$0101$, $0110$, $0111$, $1000$, $1001$, $1010$, $1011$, $1100$,
$1101$, $1110$, $1111\}$ and $\textbf{X}_2=\{2222$, $2223$,
$2232$, $2233$, $2322$, $2323$, $2332$, $2333$, $3222$, $3223$,
$3232$, $3233$, $3322$, $3323$, $3332$, $3333\}$. Clearly,
$d_{\text{S}}(\textbf{X}_{1},\textbf{X}_{2})=64=d=LM$. In fact, it
is easy to verify that for all distinct $i,j\in\{1,2,\cdots,8\}$
$d_{\text{S}}(\textbf{X}_{i},\textbf{X}_{j})=64=d=LM$.
\end{exam}

Note that by Theorem \ref{thm-bnd-S-M0}, if $d=LM$, then
$S_q(L,M,d)\leq \left\lfloor qM^{-\frac{1}{L}}\right\rfloor,$ and
so the code $\mathcal C$ constructed in Theorem
\ref{thm-extv-S-M0} is optimal with respect to code size, and
hence we have the following corollary.

\begin{cor}\label{cor-extv-S-M0}
Suppose $d=LM$ and $M^{\frac{1}{L}}$ is an integer. We have
\begin{align*}S_q(L,M,d)=
\left\lfloor qM^{-\frac{1}{L}}\right\rfloor.\end{align*}
\end{cor}

For any fixed $m\in\{2,\cdots,q\}$, we can let $d=Lm^L$, and then
we have $M^{\frac{1}{L}}=\left(\frac{d}{L}\right)^{\frac{1}{L}}=m$
is an integer. In this case, $N=\left\lfloor
qM^{-\frac{1}{L}}\right\rfloor=\left\lfloor
\frac{q}{m}\right\rfloor$.

\vspace{3pt}\begin{rem} For the case that $d=LM$ but
$M^{\frac{1}{L}}$ is not an integer, let $N$ be any fixed positive
integer such that $\left\lfloor \frac{q}{N}\right\rfloor\geq
M^{\frac{1}{L}}$. Then we can partition $\mathbb A$ into $N$
mutually disjoint subsets $W_1,W_1,\cdots,W_{N}$ such that for
each $i\in[N]$, $|W_i|\geq M^{\frac{1}{L}}$. By the same
construction as in Theorem \ref{thm-extv-S-M0}, we can obtain an
$(L,M,N,d)_q$ sequence-subset code. Thus, we have $N^*\leq
S_q(L,M,d)\leq\left\lfloor qM^{-\frac{1}{L}}\right\rfloor$, where
$N^*=\max\left\{N; \left\lfloor \frac{q}{N}\right\rfloor\geq
M^{\frac{1}{L}}\right\}$.
\end{rem}

\subsection{Construction Based on Binary Codes}
In the rest of this section, to distinguish from sequence-subset
code $($i.e., a subset of the power set $\mathcal P(\mathbb A^L)$
of the set $\mathbb A^L)$, we will call any subset of $\mathbb
A^L$ a conventional code. An $(L,N,d)_q$ \emph{conventional code}
is a subset of $\mathbb A^L$ with $N$ codewords and
\textcolor{blue}{the} minimum Hamming distance $d~($recalling that
$q$ is the size of the alphabet $\mathbb A)$. Our following
constructions of sequence-subset codes are based on conventional
codes with respect to the Hamming distance.

The following construction is a modification of Construction 2 of
\cite{Lenz18}.

\textbf{Construction 2}: Let $\mathcal C_1=\{\textbf{\text{x}}_1,
\textbf{\text{x}}_2,\cdots, \textbf{\text{x}}_K\}\subseteq \mathbb
A^L$ be a conventional code over $\mathbb A$ and $\mathcal
C_2=\{\textbf{\text{w}}_1, \textbf{\text{w}}_2,\cdots,
\textbf{\text{w}}_N\}\subseteq\mathbb F_2^{K}$ be a conventional
binary code. For each $\textbf{\text{w}}_i\in\mathcal C_2$, let
$$\textbf{\text{X}}_i=\{\textbf{\text{x}}_j;
j\in\text{supp}(\textbf{\text{w}}_i)\},$$ where
$\text{supp}(\textbf{\text{w}}_i)=\{j\in[K];
\textbf{\text{w}}_i(j)\neq 0\}$ is the support of
$\textbf{\text{w}}_i$. Finally, let $$\mathcal
C=\{\textbf{\text{X}}_1, \textbf{\text{X}}_2,\cdots,
\textbf{\text{X}}_N\}.$$

Then we have the following theorem.
\begin{thm}\label{Cntrn-BCodes}
Suppose $\mathcal C_1$ has the minimum (Hamming) distance $d_1$
and $\mathcal C_2$ has the minimum (Hamming) distance $d_2$. Then
the code $\mathcal C$ obtained by Construction 2 has sequence
length $L$, code size $|\mathcal C|=N$, and the minimum
sequence-subset distance $d_{\text{S}}(\mathcal C)$ satisfies
\begin{align*}
d_{\text{S}}(\mathcal C)\geq d_1\cdot
\left\lceil\frac{d_2}{2}\right\rceil.
\end{align*}
\end{thm}
\begin{proof}
Clearly, $\mathcal C$ has sequence length $L$ and code size
$|\mathcal C|=N$. It remains to prove that $d_{\text{S}}(\mathcal
C)\geq d_1\cdot \left\lceil\frac{d_2}{2}\right\rceil.$

Let $\textbf{\text{X}}_{i_1}$ and $\textbf{\text{X}}_{i_2}$ be any
distinct codewords of $\mathcal C$. We need to prove
$d_{\text{S}}(\textbf{\text{X}}_{i_1},
\textbf{\text{X}}_{i_2})\geq d_1\cdot
\left\lceil\frac{d_2}{2}\right\rceil.$

Without loss of generality, assume that
$|\textbf{\text{X}}_{i_1}|\leq|\textbf{\text{X}}_{i_2}|$. Then we
have $|\textbf{\text{X}}_{i_1}\backslash\textbf{\text{X}}_{i_2}|
\leq|\textbf{\text{X}}_{i_2}\backslash\textbf{\text{X}}_{i_1}|.$
To simplify notation, denote
$$\tilde{\textbf{\text{X}}}_{i_1}=
\textbf{\text{X}}_{i_1}\backslash\textbf{\text{X}}_{i_2}~~\text{and}~~
\tilde{\textbf{\text{X}}}_{i_2}=
\textbf{\text{X}}_{i_2}\backslash\textbf{\text{X}}_{i_1}.$$ For an
arbitrary injection
$\chi:\tilde{\textbf{\text{X}}}_{i_1}\rightarrow
\tilde{\textbf{\text{X}}}_{i_2}$, by \eqref{eq-def-chi-d},
\begin{align}\label{eq1-Cntrn-BCodes}
d_{\chi}(\tilde{\textbf{\text{X}}}_{i_1},
\tilde{\textbf{\text{X}}}_{i_2})\!=\!\!
\sum_{\textbf{\text{x}}\in\tilde{\textbf{\text{X}}}_{i_1}}
d_{\text{H}}(\textbf{\text{x}},
\chi(\textbf{\text{x}}))\!+\!L(|\tilde{\textbf{\text{X}}}_{i_2}|\!-
\!|\tilde{\textbf{\text{X}}}_{i_1}|).\end{align} Since $\mathcal
C_1$ has the minimum (Hamming) distance $d_1$ and by construction
of $\mathcal C$, $\textbf{\text{x}}$ and $\chi(\textbf{\text{x}})$
are distinct codeword in $\mathcal C_1$, so
$$\sum_{\textbf{\text{x}}\in\tilde{\textbf{\text{X}}}_{i_1}}
d_{\text{H}}(\textbf{\text{x}}, \chi(\textbf{\text{x}}))\geq
|\tilde{\textbf{\text{X}}}_{i_1}|\cdot d_1.$$ Moreover, since
$\mathcal C_1\subseteq \mathbb A^L$, then $L\geq d_1$. Hence,
\eqref{eq1-Cntrn-BCodes} implies that
\begin{align}\label{eq2-Cntrn-BCodes}
d_{\chi}(\tilde{\textbf{\text{X}}}_{i_1},
\tilde{\textbf{\text{X}}}_{i_2})&\geq
|\tilde{\textbf{\text{X}}}_{i_1}|\cdot
d_1+d_1(|\tilde{\textbf{\text{X}}}_{i_2}|-
|\tilde{\textbf{\text{X}}}_{i_1}|)\nonumber\\
&=d_1\cdot|\tilde{\textbf{\text{X}}}_{i_2}|\nonumber\\
&=d_1\cdot|\textbf{\text{X}}_{i_2}\backslash\textbf{\text{X}}_{i_1}|.
\end{align}

By the construction of $\mathcal C$,
$\textbf{\text{X}}_{i_1}=\{\textbf{\text{x}}_j;
j\in\text{supp}(\textbf{\text{w}}_{i_1})\}$ and
$\textbf{\text{X}}_{i_2}=\{\textbf{\text{x}}_j;
j\in\text{supp}(\textbf{\text{w}}_{i_2})\}$ for some distinct
$\textbf{\text{w}}_{i_1}, \textbf{\text{w}}_{i_2}\in\mathcal C_2$.
Then we have
\begin{align*}
|\textbf{\text{X}}_{i_1} \backslash\textbf{\text{X}}_{i_2}|
+|\textbf{\text{X}}_{i_2}\backslash\textbf{\text{X}}_{i_1}|
=d_{\text{H}}(\textbf{\text{w}}_{i_1},
\textbf{\text{w}}_{i_2})\geq d_2,\end{align*} where $d_2$ is the
minimum (Hamming) distance of $\mathcal C_2$. Note that
$|\textbf{\text{X}}_{i_1}\backslash\textbf{\text{X}}_{i_2}|
\leq|\textbf{\text{X}}_{i_2}\backslash\textbf{\text{X}}_{i_1}|.$
Then by the above equation, we have
$|\textbf{\text{X}}_{i_2}\backslash\textbf{\text{X}}_{i_1}|
\geq\frac{d_2}{2}.$ Moreover, since
$|\textbf{\text{X}}_{i_2}\backslash\textbf{\text{X}}_{i_1}|$ is an
integer, so
$$|\textbf{\text{X}}_{i_2}\backslash\textbf{\text{X}}_{i_1}|
\geq\left\lceil\frac{d_2}{2}\right\rceil.$$ Combining this with
\eqref{eq2-Cntrn-BCodes}, we have
$$d_{\chi}(\tilde{\textbf{\text{X}}}_{i_1},
\tilde{\textbf{\text{X}}}_{i_2})\geq
d_1\cdot\left\lceil\frac{d_2}{2}\right\rceil.$$ Note that
$\chi:\textbf{\text{X}}_{i_1}\backslash\textbf{\text{X}}_{i_2}
\rightarrow\textbf{\text{X}}_{i_2}\backslash\textbf{\text{X}}_{i_1}$
is an arbitrary injection. By Definition \ref{def-dst} and
Corollary \ref{cor-dst}, we have
\begin{align*}d_{\text{S}}(\textbf{\text{X}}_{i_1},\textbf{\text{X}}_{i_2})
=d_{\text{S}}(\textbf{\text{X}}_{i_1}
\backslash\textbf{\text{X}}_{i_2},\textbf{\text{X}}_{i_2}
\backslash\textbf{\text{X}}_{i_1})\geq d_1\cdot
\left\lceil\frac{d_2}{2}\right\rceil,\end{align*} which completes
the proof.
\end{proof}

\begin{exam}\label{exm-Con2} Let $\mathcal C_1$
be a binary $[5,3]$ linear code with generator matrix $G_1$ and
$\mathcal C_2$ be a binary linear $[8,3]$ code with generator
matrix $G_2$, where
\begin{eqnarray*}
G_1=\left(\begin{array}{ccccc}
1 & 0 & 0 & 1 & 0 \\
0 & 1 & 0 & 0 & 1 \\
0 & 0 & 1 & 1 & 1 \\
\end{array}\right)
\end{eqnarray*} and
\begin{eqnarray*}
G_2=\left(\begin{array}{cccccccc}
1 & 0 & 0 & 1 & 1 & 1 & 0 & 0 \\
0 & 1 & 0 & 0 & 0 & 1 & 1 & 1 \\
0 & 0 & 1 & 1 & 0 & 1 & 0 & 1 \\
\end{array}\right).
\end{eqnarray*}
We have $L=5$, $K=N=8$, the minimum distance of $\mathcal C_1$ is
$d_1=2$ and the minimum distance of $\mathcal C_2$ is $d_2=4$.
Denote $\mathcal C_1=\{\textbf{\text{x}}_1,
\textbf{\text{x}}_2,\cdots, \textbf{\text{x}}_8\}$, where
$\textbf{\text{x}}_1=00000$, $\textbf{\text{x}}_2=10010$,
$\textbf{\text{x}}_3=01001$, $\textbf{\text{x}}_4=00111$,
$\textbf{\text{x}}_5=11011$, $\textbf{\text{x}}_6=10101$,
$\textbf{\text{x}}_7=01110$, $\textbf{\text{x}}_8=11100$. Then by
Construction 2, for each codeword $\textbf{w}\in\mathcal C_2$, we
can obtain a subset $\textbf{X}_{\textbf{w}}=\{\textbf{x}_j;
j\in\text{supp}(\textbf{w})\}\subseteq\mathbb F_2^{5}$, and this
gives a code $\mathcal C=\{\textbf{X}_{\textbf{w}};
\textbf{w}\in\mathcal C_2\}$. For example, for
$\textbf{w}=10011100$, we have $\textbf{X}_{\textbf{w}}
=\{\textbf{x}_1,\textbf{x}_4,\textbf{x}_5,\textbf{x}_6\}
=\{00000,00111,11011,10101\}$; for $\textbf{w}'=11101110$, we have
$\textbf{X}_{\textbf{w}'}=\{\textbf{x}_1,\textbf{x}_2,\textbf{x}_3,
\textbf{x}_5,\textbf{x}_6,\textbf{x}_7\}=\{00000,
10010,01001,11011,10101,01110\}$. By Corollary \ref{cor-dst}, we
have $d_{\text{S}}(\textbf{\text{X}}_{\textbf{w}},
\textbf{\text{X}}_{\textbf{w}'})=
d_{\text{S}}(\textbf{\text{X}}'_{\textbf{w}},
\textbf{\text{X}}'_{\textbf{w}'})$, where
$\textbf{\text{X}}'_{\textbf{w}}=
\textbf{\text{X}}_{\textbf{w}}\backslash
\textbf{\text{X}}_{\textbf{w}'}=\{\textbf{x}_4\}=\{00111\}$ and
$\textbf{\text{X}}'_{\textbf{w}'}=
\textbf{\text{X}}_{\textbf{w}'}\backslash\textbf{\text{X}}_{\textbf{w}'}
=\{\textbf{x}_2,\textbf{x}_3,
\textbf{x}_7\}=\{10010,01001,01110\}$. Further, by
\eqref{eq-def-chi-d} and \eqref{eq-def-dst}, we can obtain
$d_{\text{S}}(\textbf{\text{X}}'_{\textbf{w}},
\textbf{\text{X}}'_{\textbf{w}'})=2+2L=12$. Note that
$|\textbf{\text{X}}_{\textbf{w}}|\neq|\textbf{\text{X}}_{\textbf{w}'}|$,
so the code $\mathcal C$ is not a constant-codeword-size code.
\end{exam}

\begin{rem}The code $\mathcal C$ obtained by Construction 2 may or
may not have constant codeword size, depending on whether
$\mathcal C_2$ is a constant weight binary code. In fact,
$\mathcal C$ is a constant-codeword-size code if and only if
$\mathcal C_2$ is a constant-weight binary code.
\end{rem}

To compare $|\mathcal C|$ with the bound in Theorem
\ref{thm-Sglnt-bnd}, we let $\mathcal C_1$ be an $(L,K,d_1)$ code
over $\mathbb A$ and $\mathcal C_2$ be a $(K,2\delta,M)$ binary
constant-weight code such that $L>d_1>\left(1-\frac{1}{q}\right)L$
and $\delta<q$. Then by construction 2 and Theorem 8, we can
obtain a sequence-subset code $\mathcal C$ with sequence length
$L$, code size $|\mathcal C|=|\mathcal C_2|$, constant codeword
size $M$, and minimum distance $d_{\text{S}}(\mathcal C)\geq
d=d_1\delta$. For $\mathcal C_1$, by the Plotking bound
\cite{Plotkin}, $$K\leq
K_0\triangleq\left\lfloor\frac{d_1}{d_1-rL}\right\rfloor,$$ where
$r=1-\frac{1}{q}$. Then we have
$$|\mathcal C_2|\leq A(K,2\delta,M)\leq A(K_0,2\delta,M),$$
where for any $n,d$ and $w$, $A(n,d,w)$ denotes the maximum number
of codewords of a binary constant weight code of length $n$,
minimum Hamming distance $d$ and constant weight $w$. By the
Johnson bound \cite{Johnson},
$$A(K_0,2\delta,M)\leq\left\lfloor\frac{K_0}{M}
\left\lfloor\frac{K_0-1}{M-1}\cdots
\left\lfloor\frac{K_0-M+\delta}{\delta}\right\rfloor\cdots
\right\rfloor\right\rfloor.$$ Then we have $$|\mathcal
C|=|\mathcal C_2|\leq\left\lfloor\frac{K_0}{M}
\left\lfloor\frac{K_0-1}{M-1}\cdots
\left\lfloor\frac{K_0-M+\delta}{\delta}\right\rfloor\cdots
\right\rfloor\right\rfloor.$$ On the other hand, since
$L>d_1>\left(1-\frac{1}{q}\right)L$ and $\delta<q$, we have
$M_0=\left\lceil\frac{d_1\delta}{L}\right\rceil=\delta$ and
$\left(1-\frac{1}{q}\right)LM_0<d_1\delta<LM_0$. According to
Theorem \ref{thm-Sglnt-bnd}, we have
\begin{align*}|\mathcal C|&\leq\left\lfloor\frac{q^L}{M}
\left\lfloor\frac{q^L-1}{M-1}\cdots
\left\lfloor\frac{q^L-M+M_0+1}{M_0+1}\frac{d_1\delta}{d_1\delta-rL\delta}
\right\rfloor\cdots\right\rfloor\right\rfloor\\
&=\left\lfloor\frac{q^L}{M} \left\lfloor\frac{q^L-1}{M-1}\cdots
\left\lfloor\frac{q^L-M+\delta+1}{\delta+1}K_0
\right\rfloor\cdots\right\rfloor\right\rfloor.\end{align*} Note
that
\begin{align*}&\left\lfloor\frac{q^L}{M} \left\lfloor\frac{q^L-1}{M-1}\cdots
\left\lfloor\frac{q^L-M+\delta+1}{\delta+1}K_0
\right\rfloor\cdots\right\rfloor\right\rfloor\\
&>\left\lfloor\frac{K_0}{M} \left\lfloor\frac{K_0-1}{M-1}\cdots
\left\lfloor\frac{K_0-M+\delta}{\delta}\right\rfloor\cdots
\right\rfloor\right\rfloor,\end{align*} which can be obtained from
the simple facts that
$K_0=\left\lfloor\frac{d_1}{d_1-rL}\right\rfloor<q^L$ and
$K_0\geq\frac{K_0-M+\delta}{\delta}~($noticing that $\delta=M_0
\leq M)$. Hence, the size of codes obtained from Construction 2
does not achieve the bound in Theorem \ref{thm-Sglnt-bnd}.

\subsection{Construction Based on $q$-ary Codes $(q\geq
2)$}

In this subsection, we present a construction based on $q$-ary
codes, where $q\geq 2$.

\textbf{Construction 3}: Let $\mathbb A$ and $\mathbb B$ be two
alphabets of size $q$ and $\tilde{q}$, respectively. Let $\mathcal
C_1$ be an $(L,M\tilde{q},d_1)_{q}$ conventional code over
$\mathbb A$ and $\mathcal C_2$ be an $(M,N,d_2)_{\tilde{q}}$
conventional code over $\mathbb B$. The $M\tilde{q}$ codewords of
$\mathcal C_1$ can be indexed as
$$\mathcal C_1=\{\textbf{x}_{i,j}: i\in[M], j\in\mathbb B\}.$$ From each
$\textbf{c}=(c_{1},c_{2},\cdots,c_{M})\in\mathcal C_2$, we can
obtain a subset
$$\textbf{X}_{\textbf{c}}=\{\textbf{x}_{1,c_1}, \textbf{x}_{2,c_2}, \cdots,
\textbf{x}_{M,c_M}\}\subseteq\mathcal C_1,$$ and finally, let
\begin{align}\label{cntrn-N-B}
\mathcal C=\{\textbf{X}_{\textbf{c}}; \textbf{c}\in\mathcal
C_2\}.\end{align}

Then $\mathcal C$ is a sequence-subset code over $\mathbb A$ and
we have the following theorem.
\begin{thm}\label{Cntrn-IdCodes}
The code $\mathcal C$ obtained by Construction 3 has sequence
length $L$, constant codeword size $M$, code size $|\mathcal
C|=N$, and minimum sequence-subset distance
\begin{align*}
d_{\text{S}}(\mathcal C)\geq d_1d_2.
\end{align*}
\end{thm}
\begin{proof}
From the construction it is easy to see that $\mathcal C$ has
sequence length $L$, constant codeword size $M$ and code size
$|\mathcal C|=N$. It remains to prove that $d_{\text{S}}(\mathcal
C)\geq d_1d_2$, that is, $d_{\text{S}}(\textbf{X}_{\textbf{c}},
\textbf{X}_{\textbf{c}'})\geq d_1d_2$ for any distinct
$\textbf{X}_{\textbf{c}}$ and $\textbf{X}_{\textbf{c}'}$ in
$\mathcal C$, where $\textbf{c}=(c_{1},c_{2},\cdots,c_{M})$ and
$\textbf{c}'=(c_{1}',c_{2}',\cdots,c_{M}')$ are any pair of
distinct codewords in $\mathcal C_2$.

Let $A$ be the set of all $i\in[M]$ such that $c_{i}\neq c_i'$.
Since $\mathcal C_2$ has the minimum (Hamming) distance $d_2$,
then
$$|A|=d_{\text{H}}(\textbf{c}, \textbf{c}')\geq d_2.$$ Denote
$$\tilde{\textbf{X}}_{\textbf{c}}=\{\textbf{x}_{i,c_i}; i\in
A\} ~~\text{and}~~
\tilde{\textbf{X}}_{\textbf{c}'}=\{\textbf{x}_{i,c_i'}; i\in
A\}.$$ Then by the construction, we have
$$\tilde{\textbf{X}}_{\textbf{c}}=\textbf{X}_{\textbf{c}}
\backslash\textbf{X}_{\textbf{c}'} ~~\text{and}~~
\tilde{\textbf{X}}_{\textbf{c}'}=\textbf{X}_{\textbf{c}'}
\backslash\textbf{X}_{\textbf{c}}.$$ By Corollary \ref{cor-dst},
it suffices to prove that
$d_{\text{S}}(\tilde{\textbf{X}}_{\textbf{c}},
\tilde{\textbf{X}}_{\textbf{c}'})\geq d_1d_2$.

Note that $|\tilde{\textbf{X}}_{\textbf{c}}|=
|\tilde{\textbf{X}}_{\textbf{c}'}|=|A|$ and
$\tilde{\textbf{X}}_{\textbf{c}}\cap
\tilde{\textbf{X}}_{\textbf{c}'}=\emptyset$. Then for any
injection $\chi:\tilde{\textbf{X}}_{\textbf{c}}\rightarrow
\tilde{\textbf{X}}_{\textbf{c}'}$, we have
\begin{align*}d_{\chi}(\tilde{\textbf{X}}_{\textbf{c}},
\tilde{\textbf{X}}_{\textbf{c}'})&=
\sum_{\textbf{x}\in\tilde{\textbf{X}}_{\textbf{c}}}d_{\text{H}}
(\textbf{x}, \chi(\textbf{x}))\\&\geq|A|\cdot d_1\\&\geq
d_1d_2,\end{align*} where the equality comes from
\eqref{eq-def-chi-d}, the first inequality comes from the
assumption that $\mathcal C_1$ has the minimum (Hamming) distance
$d_1$, and the second inequality comes from the fact that $|A|\geq
d_2$. By Definition \ref{def-dst},
$d_{\text{S}}(\tilde{\textbf{X}}_{\textbf{c}},
\tilde{\textbf{X}}_{\textbf{c}'})\geq d_1d_2$, and hence by
Corollary \ref{cor-dst}, $d_{\text{S}}(\textbf{X}_{\textbf{c}},
\textbf{X}_{\textbf{c}'})\geq d_1d_2$. Since
$\textbf{X}_{\textbf{c}}$ and $\textbf{X}_{\textbf{c}'}$ are any
pair of distinct codewords in $\mathcal C$, we have
$d_{\text{S}}(\mathcal C)\geq d_1d_2$, which completes the proof.
\end{proof}

\begin{exam}
For illustration of Construction 3, we give a simple example with
$M=4$ and $q=\tilde{q}=2$. Let $\mathcal C_1$ be the $[5,3]$
binary linear code given in Example \ref{exm-Con2} and $\mathcal
C_2=\{\textbf{c}_1, \textbf{c}_2, \textbf{c}_3, \textbf{c}_4\}$ be
a binary code, where $\textbf{c}_1=0000$, $\textbf{c}_2=0101$,
$\textbf{c}_3=1010$ and $\textbf{c}_4=1111$. The codewords of
$\mathcal C_1$ can be denoted as
$\textbf{x}_{i,j},i=1,2,3,4,j=0,1$ such that
$\textbf{x}_{1,0}=00000$, $\textbf{x}_{2,0}=01001$,
$\textbf{x}_{3,0}=10010$, $\textbf{x}_{4,0}=11011$,
$\textbf{x}_{1,1}=00111$, $\textbf{x}_{2,1}=01110$,
$\textbf{x}_{3,1}=10101$ and $\textbf{x}_{4,1}=11100$. Then by
Construction 3, we have $\mathcal C=\{\textbf{X}_{\textbf{c}_1},
\textbf{X}_{\textbf{c}_2}, \textbf{X}_{\textbf{c}_3},
\textbf{X}_{\textbf{c}_4}\}$, where
$\textbf{X}_{\textbf{c}_1}=\{\textbf{x}_{1,0}, \textbf{x}_{2,0},
\textbf{x}_{3,0}, \textbf{x}_{4,0}\}$,
$\textbf{X}_{\textbf{c}_2}=\{\textbf{x}_{1,0}, \textbf{x}_{2,1},
\textbf{x}_{3,0}, \textbf{x}_{4,1}\}$,
$\textbf{X}_{\textbf{c}_3}=\{\textbf{x}_{1,1}, \textbf{x}_{2,0},
\textbf{x}_{3,1}, \textbf{x}_{4,0}\}$ and
$\textbf{X}_{\textbf{c}_4}=\{\textbf{x}_{1,1}, \textbf{x}_{2,1},
\textbf{x}_{3,1}, \textbf{x}_{4,1}\}$. It is easy to see that
$d_2=2$ and $d_{\text{S}}(\textbf{X}_{\textbf{c}},
\textbf{X}_{\textbf{c}'})\geq d_1d_2=4$ for any distinct
$c,c'\in\mathcal C_2$. For example, by Corollary \ref{cor-dst},
$d_{\text{S}}(\textbf{X}_{\textbf{c}_1},
\textbf{X}_{\textbf{c}_2})=d_{\text{S}}(\{\textbf{x}_{2,0},
\textbf{x}_{4,0}\}, \{\textbf{x}_{2,1},\textbf{x}_{4,1}\})=6$.
\end{exam}

The following is a more general example of Construction 3.

\begin{exam}\label{exm-Con3}
Let $\mathcal C_1$ be an $[L,k,d_1]_q$ linear code such that the
first $k$ symbols of the codewords of $\mathcal C_1$ are the
information symbols. For any given integer $r$ such that $1\leq
r<k$, let $\tilde{q}=q^r$ and $M=q^{s}$, where $s=k-r$. Note that
there exists a bijection $\pi:[M]\rightarrow\mathbb F_q^{s}$.
Moreover, fixing a basis, each element of $\mathbb F_{q^r}$ can be
uniquely represented as a vector in $\mathbb F_q^r$, so we can
identify each element of $\mathbb F_{q^r}$ as a vector in $\mathbb
F_q^r$. Then for each $i\in[M]$ and each $j\in\mathbb F_{q^r}$, we
can let
\begin{align*}
&~\textbf{x}_{i,j}=(x_1, x_2, \cdots, x_L):\\
&(x_1, x_2, \cdots, x_s)=\pi(i) ~~\text{and}~~ (x_{s+1}, \cdots,
x_k)=j.\end{align*} Now, let $\mathcal C_2$ be an
$[M,K,d_2]_{q^r}$ linear code, where $K\in[M]$ is another design
parameter. Then for each
$\textbf{c}=(c_{1},c_{2},\cdots,c_{M})\in\mathcal C_2$, we can
obtain
$$\textbf{X}_{\textbf{c}}=\{\textbf{x}_{1,c_1}, \textbf{x}_{2,c_2}, \cdots,
\textbf{x}_{M,c_M}\}\subseteq\mathcal C_1,$$ that is, for each
$i\in[M]$, $\textbf{x}_{i,c_i}=(x_1, x_2, \cdots, x_L)$ such that
$$(x_1, x_2, \cdots, x_s)=\pi(i) ~~\text{and}~~ (x_{s+1}, \cdots,
x_k)=c_i.$$ Finally, we have
\begin{align*}
\mathcal C=\{\textbf{X}_{\textbf{c}}; \textbf{c}\in\mathcal
C_2\}.\end{align*} By Theorem \ref{Cntrn-IdCodes}, $\mathcal C$ is
a sequence-subset code of sequence length $L$, constant codeword
size $M=q^s$, code size $|\mathcal C|=|\mathcal C_2|=q^{rK}$, and
minimum sequence-subset distance $d_{\text{S}}(\mathcal C)\geq
d_1d_2.$

As a special case of Example \ref{exm-Con3}, suppose $q\geq L-1$
and $s\leq\frac{k}{2}$. We can let $\mathcal C_1$ be an
$[L,k,L-k+1]_q$ MDS code and $\mathcal C_2$ be a
$[q^s,K,q^s-K+1]_{q^r}$ MDS code. Then the minimum distance of
$\mathcal C$ satisfies $d_{\text{S}}(\mathcal C)\geq
(L-k+1)(q^s-K+1).$ This special case is essentially similar to the
method used in \cite{Grass15}.
\end{exam}

Let $\mathcal C$ be a sequence-subset code obtained by
Construction 3. To compare $|\mathcal C|$ with the bound in
Theorem \ref{thm-pltk-bnd}, we consider $\tilde{q}>q$,
$L>d_1>\left(1-\frac{1}{q}\right)\left(1-\frac{1}{\tilde{q}}\right)^{-1}L$
and $M>d_2>(1-\frac{1}{\tilde{q}})M$. Then
$L>d_1>\left(1-\frac{1}{q}\right)L$ and
$LM>d_1d_2>\left(1-\frac{1}{q}\right)LM$. By the Plotking bound
\cite{Plotkin}, we have
\begin{align}\label{eq1-com-Con3}M\tilde{q}=|\mathcal
C_1|\leq\frac{d_1}{d_1-rL},\end{align} where $r=1-\frac{1}{q},$
and
\begin{align}\label{eq2-com-Con3}
N=|\mathcal C_2|\leq\frac{d_2}{d_2-\tilde{r}M},\end{align} where
$\tilde{r}=1-\frac{1}{\tilde{q}},$ On the other hand, by Theorem
\ref{thm-pltk-bnd}, \begin{align}\label{eq3-com-Con3} |\mathcal
C|=N\leq\frac{d_1d_2}{d_1d_2-rLM}.\end{align} By
\eqref{eq1-com-Con3}, we have
$$\frac{rL}{d_1}>1-\frac{1}{M\tilde{q}}\geq 1-\frac{1}{\tilde{q}}=\tilde{r},$$
which implies that
$$1-\frac{\tilde{r}M}{d_2}>1-\frac{rLM}{d_1d_2},$$ and so
$$\frac{d_2}{d_2-\tilde{r}M}<\frac{d_1d_2}{d_1d_2-rLM}.$$ By
\eqref{eq2-com-Con3} and \eqref{eq3-com-Con3}, $|\mathcal C|$ does
not achieve the bound in Theorem \ref{thm-pltk-bnd}.

\subsection{Construction Based on Sequence Index}
In this subsection, if $\textbf{\text{x}}=(\textbf{\text{x}}(1),
\textbf{\text{x}}(2), \cdots, \textbf{\text{x}}(L))\in\mathbb A^L$
and $I=\{i_1,i_2,\cdots,i_m\}\subseteq[L]$ such that
$i_1<i_2<\cdots<i_m$, then we denote
$\textbf{\text{x}}(I)=(\textbf{\text{x}}(i_1),
\textbf{\text{x}}(i_2), \cdots, \textbf{\text{x}}(i_m))$.

The construction given in this subsection is a slight improvement
of the Construction 1 of \cite{Lenz18}.

\textbf{Construction 4}: Let $\mathcal C_1=\{\textbf{s}_1,
\textbf{s}_2,\cdots, \textbf{s}_M\}\subseteq\mathbb A^{L_1}$ be a
conventional code over $\mathbb A$ with block length $L_1$ and
\textcolor{blue}{the} minimum (Hamming) distance $d_1$, and
$\mathcal C_2=\{\textbf{u}_1, \textbf{u}_2,\cdots,
\textbf{u}_N\}\subseteq\mathbb A^{d_1M}$ be a conventional code
over $\mathbb A$ with block length $d_1M$ and
\textcolor{blue}{the} minimum (Hamming) distance $d_2$. For each
$j\in[M]$, let
$$I_j=\{\ell\in\mathbb Z; (j-1)d_1<\ell\leq jd_1\}$$
and for each $i\in[N]$, let
\begin{align*}
\textbf{X}_{i}= \{\textbf{x}_{i,1}, \textbf{x}_{i,2}, \cdots,
\textbf{x}_{i,M}\}\end{align*} such that for each $j\in[M]$,
$$\textbf{x}_{i,j}=(\textbf{\text{s}}_j,
\textbf{u}_{i}(I_j)).$$ Finally, let
\begin{align}\label{eq-cntrn-Id}
\mathcal C=\{\textbf{\text{X}}_{i}; i\in[N]\}.
\end{align}

Then $\mathcal C$ is a sequence-subset code over $\mathbb A$. In
this construction, each codeword $\textbf{s}_j$ of $\mathcal C_1$
serves as an index of the sequence $\textbf{x}_{i,j}$ of the
codeword $\textbf{X}_{i}$, and $\textbf{u}_{i}(I_j)$ is the
information part of $\textbf{x}_{i,j}$. For this reason, this
construction is called a construction based on sequence index.
Moreover, we have the following theorem.
\begin{thm}\label{Cntrn-IdCodes}
The code $\mathcal C$ obtained by Construction 4 has sequence
length $L=L_1+d_1$, constant codeword size $M$, code size
$|\mathcal C|=N$, and minimum sequence-subset distance
\begin{align*}
d_{\text{S}}(\mathcal C)\geq d_2.
\end{align*}
\end{thm}
\begin{proof}
Clearly, $\mathcal C$ has sequence length $L=L_1+d_1$, constant
codeword size $M$ and code size $|\mathcal C|=N$. It remains to
prove that $d_{\text{S}}(\mathcal C)\geq d_2.$

Let $i_1, i_2\in[N]$ be any two distinct elements of $[N]$, we
need to prove that $d_{\text{S}}(\textbf{X}_{i_1},
\textbf{X}_{i_2})\geq d_2,$ where $\textbf{X}_{i_1}=
\{\textbf{x}_{i_1,1}, \textbf{x}_{i_1,2}, \cdots,
\textbf{x}_{i_1,M}\}$ and $\textbf{X}_{i_2}= \{\textbf{x}_{i_2,1},
\textbf{x}_{i_2,2}, \cdots, \textbf{x}_{i_2,M}\}$. For any
permutation\footnote{Note that any bijection between
$\textbf{X}_{i_1}$ and $\textbf{X}_{i_2}$ can be uniquely
represented by a permutation on the index set $[M]$, so when
applying \eqref{eq-def-chi-d} to the pair
$\{\textbf{X}_{i_1},\textbf{X}_{i_2}\}$, we can use permutations
on $[M]$ to replace bijections between $\textbf{X}_{i_1}$ and
$\textbf{\text{X}}_{i_2}$.} $\chi:[M] \rightarrow[M]$, let
$$\mathcal N=\{j\in[M];\chi(j)=j\}$$ and
$$\tilde{\mathcal N}=\{j\in[M];\chi(j)
\neq j\}.$$ Then $\mathcal N\cap\tilde{\mathcal N}=\emptyset$ and
$\mathcal N\cup\tilde{\mathcal N}=[M]$. Moreover, by
\eqref{eq-def-chi-d}, we have
\begin{align}\label{eq1-Cntrn-IdCodes}
d_{\chi}(\textbf{X}_{i_1}, \textbf{X}_{i_2})&=\!
\sum_{j=1}^Md_{\text{H}}(\textbf{x}_{i_1,j},
\textbf{x}_{i_2,\chi(j)})\nonumber\\
&=\!\sum_{j\in\mathcal N}d_{\text{H}}(\textbf{x}_{i_1,j},
\textbf{x}_{i_2,\chi(j)})+\!\sum_{j\in\tilde{\mathcal
N}}d_{\text{H}}(\textbf{x}_{i_1,j},
\textbf{x}_{i_2,\chi(j)}\nonumber\\
&=\!\sum_{j\in\mathcal N}d_{\text{H}}(\textbf{x}_{i_1,j},
\textbf{x}_{i_2,j})+\!\sum_{j\in\tilde{\mathcal
N}}d_{\text{H}}(\textbf{x}_{i_1,j},
\textbf{x}_{i_2,\chi(j)}.\end{align} We will estimate the two
terms of the right side of Equation \eqref{eq1-Cntrn-IdCodes}
separately.

First, by the construction, we have
\begin{align*}\sum_{j=1}^M\!d_{\text{H}}(\textbf{\text{x}}_{i_1,j},
\textbf{\text{x}}_{i_2,j})
&=\sum_{j=1}^M\!d_{\text{H}}(\textbf{\text{u}}_{i_1}(I_j),
\textbf{\text{u}}_{i_2}(I_j))\\&=d_{\text{H}}(\textbf{\text{u}}_{i_1},
\textbf{\text{u}}_{i_2})=d_2.\end{align*} Moreover, since for each
$i\in[N]$ and $j\in[M]$, $\textbf{\text{u}}_{i}(I_j)$ has length
$d_1$, then again by construction of $\mathcal C$, we have
\begin{align*}d_{\text{H}}(\textbf{\text{x}}_{i_1,j},
\textbf{\text{x}}_{i_2,j})
=d_{\text{H}}(\textbf{\text{u}}_{i_1}(I_j),
\textbf{\text{u}}_{i_2}(I_j))\leq d_1.\end{align*} Hence, we
obtain
\begin{align*}\sum_{j\in\mathcal N}
\!d_{\text{H}}(\textbf{\text{x}}_{i_1,j},
\textbf{\text{x}}_{i_2,j})\!
&=\!\sum_{j=1}^M\!d_{\text{H}}(\textbf{\text{x}}_{i_1,j},
\textbf{\text{x}}_{i_2,j})-\!\sum_{j\in\tilde{\mathcal
N}}\!d_{\text{H}}(\textbf{\text{x}}_{i_1,j},
\textbf{\text{x}}_{i_2,j})\\
&=\!\sum_{j=1}^M\!d_{\text{H}}(\textbf{\text{u}}_{i_1}(I_j),
\textbf{\text{u}}_{i_2}(I_j))\\&~ ~ ~ -\!\sum_{j\in\tilde{\mathcal
N}}\!d_{\text{H}}(\textbf{\text{u}}_{i_1}(I_j),
\textbf{\text{u}}_{i_2}(I_j))\\&\geq d_2-|\tilde{\mathcal N}|\cdot
d_1.\end{align*}

Second, since $\mathcal C_1$ has the minimum (Hamming) distance
$d_1$, then by construction of $\mathcal C$, we have
\begin{align*}\sum_{j\in\tilde{\mathcal
N}}d_{\text{H}}(\textbf{\text{x}}_{i_1,j},
\textbf{\text{x}}_{i_2,\chi(j)} \geq\!\sum_{j\in\tilde{\mathcal
N}}\!d_{\text{H}}(\textbf{\text{s}}_{j},
\textbf{\text{s}}_{\chi(j)})\geq|\tilde{\mathcal N}|\cdot
d_1.\end{align*}

Combining the above two inequalities with
\eqref{eq1-Cntrn-IdCodes}, we obtain
\begin{align*}
d_{\chi}(\textbf{\text{X}}_{i_1},
\textbf{\text{X}}_{i_2})&=\!\sum_{j\in\mathcal
N}d_{\text{H}}(\textbf{\text{x}}_{i_1,j},
\textbf{\text{x}}_{i_2,j})+\!\sum_{j\in\tilde{\mathcal
N}}d_{\text{H}}(\textbf{\text{x}}_{i_1,j},
\textbf{\text{x}}_{i_2,\chi(j)}\\&\geq \left(d_2-|\tilde{\mathcal
N}|\cdot d_1\right)+|\tilde{\mathcal N}|\cdot
d_1\\&=d_2.\end{align*}

Since $\chi:[M]\rightarrow [M]$ is an arbitrary bijection, then by
Definition \ref{def-dst}, we have
\begin{align*}d_{\text{S}}(\textbf{\text{X}}_{i_1},
\textbf{\text{X}}_{i_2})\geq d_2.\end{align*} Moreover, since
$i_1$ and $i_2$ are any two distinct elements of $[N]$, so we have
\begin{align*} d_{\text{S}}(\mathcal C)\geq d_2,
\end{align*} which completes the proof.
\end{proof}

\begin{exam}\label{exam-Cons4} Let $\mathbb
A=\mathbb F_2$, $\mathcal C_1=\{\textbf{s}_1, \textbf{s}_2\}$ and
$\mathcal C_2=\{\textbf{u}_1, \textbf{u}_2, \textbf{u}_3\}$, where
$\textbf{s}_1=0000$, $\textbf{s}_2=1111$, $\textbf{u}_1=00000000$,
$\textbf{u}_2=11111000$ and $\textbf{u}_3=01010111$. We can check
that $L_1=d_1=4$, $M=2$, $N=3$ and $d_2=5$. We can divide each
$\textbf{u}_i$ into $M=2$ segments, each of length $L_1=4$, and
denote each $\textbf{u}_i=\textbf{u}_{i,1}\textbf{u}_{i,2}$. For
example, $\textbf{u}_{3,1}=0101$ and $\textbf{u}_{3,2}=0111$. By
Construction 4,
we can obtain $\mathcal C=\{\textbf{X}_1, \textbf{X}_2,
\textbf{X}_3\}$, where
$\textbf{X}_1=\{\textbf{s}_1\textbf{u}_{1,1},
\textbf{s}_2\textbf{u}_{1,2}\}=\{00000000,11110000\}$,
$\textbf{X}_2=\{\textbf{s}_1\textbf{u}_{2,1},
\textbf{s}_2\textbf{u}_{2,2}\}=\{00001111,11111000\}$ and
$\textbf{X}_3=\{\textbf{s}_1\textbf{u}_{3,1},
\textbf{s}_2\textbf{u}_{3,2}\}=\{00000101,11110111\}$. It is easy
to verify that $d_{\text{S}}(\mathcal C)=5=d_2$.
\end{exam}

In Construction 1 of \cite{Lenz18}, each sequence
$\textbf{x}_{i,j}=(\textbf{\text{s}}_j, \textbf{u}_{i,j})$ such
that each $\textbf{u}_{i,j}$ is viewed as an element of the field
$\mathbb F_{q^{L-\lceil\log M\rceil}}$ and
$(\textbf{u}_{i,1},\cdots,\textbf{u}_{i,M})$ is a codeword of an
MDS code over $\mathbb F_{q^{L-\lceil\log M\rceil}}$. In
comparison, our construction uses codes $($i.e., $\mathcal C_2)$
of length $d_1M$ over $\mathbb F_{q}$ (rather than its extension
field), and $\textbf{u}_{i,1},\cdots,\textbf{u}_{i,M}$ are
obtained by dividing each codeword of $\mathcal C_2$ into $M$
segments of length $d_1$, which allows us to construct $\mathcal
C_2$ with greater sequence-subset distance. For example, suppose
$q=2$, $d_1=8$ $M=10$ and $N=2^{48}$. Then by Construction 1 of
\cite{Lenz18}, we need a $[10,6]$ MDS codes over the field
$\mathbb F_{2^8}$, which has minimum distance $5$. In comparison,
by our construction, we can let $\mathcal C_2$ be a $[80,48]$
linear code over $\mathbb F_2$ with minimum distance
$d_2=10~($e.g., see \cite{Grassl}$)$. By Theorem 10, the
corresponding sequence-subset code has a greater minimum distance.

Construction 4 can be extended to the following construction.

\textbf{Construction 4$'$}: Let $\mathcal C_1$, $\mathcal C_2$ and
$I_j, j\in[M],$ be the same as in Construction 4. Let $n$ be a
given positive integer. For each $n$-tuple
$\textbf{\text{i}}=(i_1,i_2,\cdots,i_{n})\in[N]^{n}$, let
\begin{align*}
\textbf{\text{X}}_{\textbf{\text{i}}}=
\{\textbf{\text{x}}_{\textbf{\text{i}},1},
\textbf{\text{x}}_{\textbf{\text{i}},2}, \cdots,
\textbf{\text{x}}_{\textbf{\text{i}},M}\}\end{align*} such that
for each $j\in[M]$,
$$\textbf{\text{x}}_{\textbf{\text{i}},j}=(\textbf{\text{s}}_j,
\textbf{\text{u}}_{i_1}\!(I_j), \cdots\!,
\textbf{\text{u}}_{i_{n}}\!(I_j)).$$ Finally, let
\begin{align*}
\tilde{\mathcal C}=\{\textbf{\text{X}}_{\textbf{\text{i}}};
\textbf{\text{i}}=(i_1,i_2,\cdots,i_{n})\in[N]^{n}\}.\end{align*}

Then we have the following theorem.
\begin{thm}
The code $\tilde{\mathcal C}$ obtained by Construction 4$'$ has
sequence length $L=L_1+nd_1$, constant codeword size $M$ and code
size $|\tilde{\mathcal C}|=N^{n}$, and minimum sequence-subset
distance
\begin{align*}
d_{\text{S}}(\tilde{\mathcal C})\geq d_2.
\end{align*}
\end{thm}
\begin{proof}
Clearly, the code $\tilde{\mathcal C}$ has sequence length
$L=L_1+nd_1$, constant codeword size $M$ and code size
$|\tilde{\mathcal C}|=N^{n}$. To estimate the minimum distance of
$\tilde{\mathcal C}$, suppose
$(i_1,i_2,\cdots,i_{n})=\textbf{\text{i}}\neq
\textbf{\text{i}}'=(i_1',i_2',\cdots,i_{n}')\in[N]^{n}$. Without
loss of generality, assume $i_1\neq i_1'$. For each $j\in[M]$,
consider the subsequence
$\textbf{\text{x}}'_{\textbf{\text{i}},j}=(\textbf{\text{s}}_j,
\textbf{\text{u}}_{i_1}\!(I_j))$ of
$\textbf{\text{x}}_{\textbf{\text{i}},j}$ and the subsequence
$\textbf{\text{x}}'_{\textbf{\text{i}}',j}=(\textbf{\text{s}}_j,
\textbf{\text{u}}_{i_1'}\!(I_j))$ of
$\textbf{\text{x}}_{\textbf{\text{i}}',j}$. By the same
discussions as in the proof of Theorem \ref{Cntrn-IdCodes}, we can
prove that $d_{\text{S}}(\textbf{\text{X}}'_{\textbf{i}},
\textbf{\text{X}}'_{\textbf{i}'})\geq d_2,$ where
$\textbf{\text{X}}'_{\textbf{i}}=
\{\textbf{\text{x}}'_{\textbf{\text{i}},j}; j\in[M]\}$ and
$\textbf{\text{X}}'_{\textbf{i}'}=
\{\textbf{\text{x}}'_{\textbf{\text{i}}',j}; j\in[M]\}$. Since
sequences in $\textbf{\text{X}}'_{\textbf{i}}~($resp.
$\textbf{\text{X}}'_{\textbf{i}'})$ are subsequences of
$\textbf{\text{X}}_{\textbf{i}}~($resp.
$\textbf{\text{X}}_{\textbf{i}'})$, then it is easy to see that
$d_{\text{S}}(\textbf{\text{X}}_{\textbf{i}},
\textbf{\text{X}}_{\textbf{i}'})\geq
d_{\text{S}}(\textbf{\text{X}}'_{\textbf{i}},
\textbf{\text{X}}'_{\textbf{i}'})\geq d_2.$ Hence, the minimum
sequence-subset distance of $\tilde{\mathcal C}$ satisfies
$d_{\text{S}}(\tilde{\mathcal C})\geq d_2$.
\end{proof}

\begin{exam} Let $\mathcal C_1$ and $\mathcal
C_2$ be the same as in Example \ref{exam-Cons4}, and let $n=2$. As
in Example \ref{exam-Cons4}, we can denote each $\textbf{u}_i$ as
$\textbf{u}_i=\textbf{u}_{i,1}\textbf{u}_{i,2}$. By Construction
4$'$, for each $\textbf{i}=(i_1,i_2)\in\{1,2,3\}^2$, we have
$\textbf{X}_{\textbf{i}}=
\{\textbf{s}_1\textbf{u}_{i_1,1}\textbf{u}_{i_2,1},
\textbf{s}_2\textbf{u}_{i_1,2}\textbf{u}_{i_2,2}\}$. For example,
for $\textbf{i}=(1,1)$,
$\textbf{X}_{(1,1)}=\{000000000000,111100000000\}$; for
$\textbf{i}=(1,2)$,
$\textbf{X}_{(1,2)}=\{000000001111,111100001000\}$; for
$\textbf{i}=(3,2)$,
$\textbf{X}_{(3,2)}=\{000001011111,111101111000\}$. We can
estimate $d_{\text{S}}(\textbf{X}_{(1,2)}, \textbf{X}_{(3,2)})$ as
follows. Consider $d_{\text{S}}(\textbf{\text{X}}'_{(1,2)},
\textbf{\text{X}}'_{(3,2)})$, where
$\textbf{X}'_{(1,2)}=\{\textbf{s}_1\textbf{u}_{2,1},
\textbf{s}_2\textbf{u}_{2,2}\}=\{00000000,11110000\}$ is a
subsequence of $\textbf{\text{X}}_{(1,2)}$ and
$\textbf{X}'_{(3,2)}=\{\textbf{s}_1\textbf{u}_{3,1},
\textbf{s}_2\textbf{u}_{3,2}\}=\{00000101,11110111\}$ is a
subsequence of $\textbf{\text{X}}_{(3,2)}$. Note that
$\textbf{\text{X}}'_{(1,2)}$ and $\textbf{\text{X}}'_{(3,2)}$ are
two distinct codewords of the code constructed in Example
\ref{exam-Cons4}, so we have
$d_{\text{S}}(\textbf{\text{X}}'_{(1,2)},
\textbf{\text{X}}'_{(3,2)})\geq 5$, and hence
$d_{\text{S}}(\textbf{\text{X}}_{(1,2)},
\textbf{\text{X}}_{(3,2)})\geq
d_{\text{S}}(\textbf{\text{X}}'_{(1,2)},
\textbf{\text{X}}'_{(3,2)})\geq 5$. Similarly, we can verify that
$d_{\text{S}}(\textbf{X}_{\textbf{i}},
\textbf{X}_{\textbf{i}'})\geq 5$ for all distinct
$\textbf{i},\textbf{i}'\in\{1,2,3\}^2$. Thus, the minimum
sequence-subset distance of $\tilde{\mathcal C}$ satisfies
$d_{\text{S}}(\tilde{\mathcal C})\geq 5=d_2$, where
$\tilde{\mathcal
C}=\{\textbf{X}_{\textbf{i}};\textbf{i}\in\{1,2,3\}^2\}$. In fact,
we have $d_{\text{S}}(\tilde{\mathcal C})=5$ because we can verify
that $d_{\text{S}}(\textbf{\text{X}}_{(1,2)},
\textbf{\text{X}}_{(3,2)})=5$.
\end{exam}

\section{Conclusions and Discussions}

We introduced a new metric over the power set of the set of all
vectors over a finite alphabet, which generalizes the classical
Hamming distance and was used to establish a uniform framework to
design error-correcting codes for DNA storage channel. Some upper
bounds on the size of the sequence-subset codes were derived and
some constructions of such codes were proposed.

\subsection{Open Problems in Sequence-subset Codes}
It is still an open problem to analyze the tight upper bound on
the size of sequence-subset codes and design optimal codes for
general parameters of sequence length, codeword size and minimum
distance. Another interesting problem is how to design
sequence-subset codes for DNA storage channel that can be
efficiently encoded and decoded.

\subsection{Sequence-Subset Distance for Multisets}

The sequence-subset distance $($Definition \ref{def-dst}$)$ can be
directly generalized to multisets of sequences in $\mathbb A^L$.
The following is an example of sequence-subset distance between
multisets.

\begin{exam}
Suppose $\mathbb A=\{0,1\}$ and $L=4$. Consider
$\textbf{X}_1=\{\textbf{x}_1,\textbf{x}_2,\textbf{x}_3\}$ and
$\textbf{X}_2=\{\textbf{y}_1,\textbf{y}_2,\textbf{y}_3,\textbf{y}_4\}$,
where $\textbf{x}_1=\textbf{x}_2=0101$, $\textbf{x}_3=1011$,
$\textbf{y}_1=0111$, $\textbf{y}_2=1101$, and
$\textbf{y}_3=\textbf{y}_4=1001$. Let
$\chi_0:\textbf{X}_1\rightarrow\textbf{X}_2$ be such that
$\chi_0(\textbf{x}_i)=\textbf{y}_i, i=1,2,3$. Then we can obtain
$d_{\text{H}}(\textbf{x}_i,\chi_0(\textbf{x}_i))=1$ for all
$i\in\{1,2,3\}$, and by \eqref{eq-def-chi-d}, we have
$d_{\chi_0}(\textbf{X}_1,\textbf{X}_2)=7$. Note that
$d_{\text{H}}(\textbf{x}_i,\textbf{y}_j)\geq 1$ for all
$\textbf{x}_i\in\textbf{X}_1$ and $\textbf{y}_j\in\textbf{X}_2$.
Then we have $d_{\chi}(\textbf{X}_1,\textbf{X}_2)\geq 7$ for all
$\chi\in\mathscr X$, and hence by \eqref{eq-def-dst},
$d_{\text{S}}(\textbf{X}_1,\textbf{X}_2)
=d_{\chi_0}(\textbf{X}_1,\textbf{X}_2)=7$.
\end{exam}

By similar discussions as in Section
\uppercase\expandafter{\romannumeral 2}.A, we can prove that the
function $d_{\text{S}}(\textbf{X},\textbf{Y})$ is a distance
function, where $\textbf{X}$ and $\textbf{Y}$ are multisets of
sequences in $\mathbb A^L$. Using sequence-subset distance between
multisets, we can allow the output of the DNA storage channel to
be multisets (rather than sets) of sequences in $\mathbb A^L$.
Since in the real DNA storage, some DNA strands may have many
copies that are sequenced, multisets are more suitable than sets
for the output of the DNA storage channel.

Another advantage of using multisets as the output of the DNA
storage channel is that it captures the case when there are
$t~(>1)$ strands that are changed to the same strand by
substitution errors. Note that if using sets as the output of the
channel, then $t-1$ of these sequences have to be viewed as lost
sequences, which induces a larger sequence-subset distance between
the input and output.

To study the properties of codes over the space of all multisets
of $\mathbb A^L$ with sequence-subset distance is also a possible
research direction.

\appendices

\section{Proof of Lemma \ref{lem-dst}}
If $\textbf{\text{X}}_1\cap\textbf{\text{X}}_2=\emptyset$, the
claim is naturally true. In the following, we assume that
$\textbf{\text{X}}_1\cap\textbf{\text{X}}_2\neq\emptyset$.

First, we claim that for each $\chi\in\mathscr X$ such that
$d_{\text{S}}(\textbf{\text{X}}_1,\textbf{\text{X}}_2)=
d_{\chi}(\textbf{\text{X}}_1,\textbf{\text{X}}_2)$ and each
$\textbf{y}\in\textbf{X}_1\cap\textbf{X}_2$, there exists an
$\textbf{\text{x}}\in\textbf{\text{X}}_1$ such that
$\textbf{\text{y}}=\chi(\textbf{\text{x}})$. This can be proved,
by contradiction, as follows. Suppose there is a
$\textbf{\text{y}}\in\textbf{\text{X}}_1\cap\textbf{\text{X}}_2$
such that $\textbf{\text{y}}\neq\chi(\textbf{\text{x}}')$ for all
$\textbf{\text{x}}'\in\textbf{\text{X}}_1$. Since
$\textbf{y}\in\textbf{X}_1\cap\textbf{X}_2$, then we have
$\chi(\textbf{y})\neq\textbf{\text{y}}$, and hence we can let
$\chi':\textbf{\text{X}}_1\rightarrow\textbf{\text{X}}_2$ be such
that $\chi'(\textbf{\text{y}})=\textbf{\text{y}}$ and
$\chi'(\textbf{\text{x}}')=\chi(\textbf{\text{x}}')$ for all
$\textbf{\text{x}}'\in\textbf{\text{X}}_1\backslash\{\textbf{\text{y}}\}$
(see Fig. \ref{map-1} for an illustration). Note that
$d_{\text{H}}(\textbf{\text{y}},
\chi'(\textbf{\text{y}}))=0<d_{\text{H}}(\textbf{\text{y}},
\chi(\textbf{\text{y}}))$ and $d_{\text{H}}(\textbf{\text{x}}',
\chi'(\textbf{\text{x}}'))=d_{\text{H}}(\textbf{\text{x}}',
\chi(\textbf{\text{x}}'))$ for all
$\textbf{\text{x}}'\in\textbf{\text{X}}_1\backslash\{\textbf{\text{y}}\}.$
By \eqref{eq-def-chi-d}, we have
$d_{\chi'}(\textbf{\text{X}}_1,\textbf{\text{X}}_2)<
d_{\chi}(\textbf{\text{X}}_1,\textbf{\text{X}}_2)$, which
contradicts to \eqref{eq-def-dst}. Hence, by contradiction, for
each
$\textbf{\text{y}}\in\textbf{\text{X}}_1\cap\textbf{\text{X}}_2$,
there exists an $\textbf{\text{x}}\in\textbf{\text{X}}_1$ such
that $\textbf{\text{y}}=\chi(\textbf{\text{x}})$.

\renewcommand\figurename{Fig}
\begin{figure}[htbp]
\begin{center}
\includegraphics[height=3.0cm]{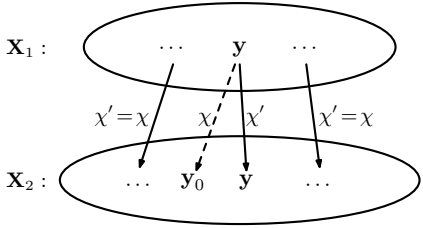}
\end{center}
\vspace{-0.2cm}\caption{An illustration of the injections in the
proof of Lemma \ref{lem-dst}: For the injection $\chi$, there
exists a
$\textbf{\text{y}}\in\textbf{\text{X}}_1\cap\textbf{\text{X}}_2$
such that $\chi(\textbf{\text{y}})\neq\textbf{\text{y}}$. Denote
$\chi(\textbf{\text{y}})=\textbf{\text{y}}_0$. Then we can modify
the injection $\chi$ to a different injection $\chi'$ by letting
$\chi'(\textbf{\text{y}})=\textbf{\text{y}}$, and the image of all
other elements of $\textbf{\text{X}}_1$ keep
unchanged.}\label{map-1}
\end{figure}
\renewcommand\figurename{Fig}
\begin{figure}[htbp]
\begin{center}
\includegraphics[height=3.0cm]{dnafig1.3}
\end{center}
\vspace{-0.2cm}\caption{An illustration of the bijections in the
proof of Lemma \ref{lem-dst}: For the bijection $\chi$, we have
$\chi(\textbf{\text{x}})=\textbf{\text{y}}$ and
$\chi(\textbf{\text{y}})=\textbf{\text{y}}_0\neq\textbf{\text{y}}$,
where
$\textbf{\text{y}}\in\textbf{\text{X}}_1\cap\textbf{\text{X}}_2$.
We modify the bijection $\chi$ to a different bijection $\chi'$ by
letting $\chi'(\textbf{\text{x}})=\textbf{\text{y}}_0$ and
$\chi'(\textbf{\text{y}})=\textbf{\text{y}}$, and the image of all
other elements of $\textbf{\text{X}}_1$ keeping
unchanged.}\label{map-2}
\end{figure}

Now, pick a $\chi\in\mathscr X$ such that
$d_{\text{S}}(\textbf{\text{X}}_1,\textbf{\text{X}}_2)=
d_{\chi}(\textbf{\text{X}}_1,\textbf{\text{X}}_2)$ and denote
$$\mathcal
N(\chi)=\{\textbf{\text{y}}'\in\textbf{\text{X}}_1\cap\textbf{\text{X}}_2;
\chi(\textbf{\text{y}}')\neq\textbf{\text{y}}'\}.$$ If $\mathcal
N(\chi)=\emptyset$, then by the definition of $\mathcal N(\chi)$,
$\chi(\textbf{\text{x}})=\textbf{\text{x}}$ for all
$\textbf{\text{x}}\in\textbf{\text{X}}_1\cap\textbf{\text{X}}_2$
and we can choose $\chi_0=\chi$. Otherwise, pick a
$\textbf{\text{y}}\in\mathcal N(\chi)$ and we have
$\chi(\textbf{\text{y}})=\textbf{\text{y}}_0$ for some
$\textbf{\text{y}}_0\in\textbf{\text{X}}_2\backslash\{\textbf{\text{y}}\}$.
Moreover, by previous discussion, there exists an
$\textbf{\text{x}}\in\textbf{\text{X}}_1$ such that
$\textbf{\text{y}}=\chi(\textbf{\text{x}})$. Then we can let
$\chi':\textbf{\text{X}}_1\rightarrow\textbf{\text{X}}_2$ be such
that $\chi'(\textbf{\text{x}})=\textbf{\text{y}}_0$,
$\chi'(\textbf{\text{y}})=\textbf{\text{y}}$ and
$\chi'(\textbf{\text{x}}')=\chi_0(\textbf{\text{x}}')$ for all
$\textbf{\text{x}}'\in\textbf{\text{X}}_1\backslash\{\textbf{\text{x}},
\textbf{\text{y}}\}$ (see Fig. \ref{map-2} for an illustration).
Note that
\begin{align*}d_{\text{H}}(\textbf{\text{x}},
\chi'(\textbf{\text{x}}))+d_{\text{H}}(\textbf{\text{y}},
\chi'(\textbf{\text{y}}))&=d_{\text{H}}(\textbf{\text{x}},
\textbf{\text{y}}_0)+d_{\text{H}}(\textbf{\text{y}},
\textbf{\text{y}})\\&=d_{\text{H}}(\textbf{\text{x}},
\textbf{\text{y}}_0)\\&\leq d_{\text{H}}(\textbf{\text{x}},
\textbf{\text{y}})+d_{\text{H}}(\textbf{\text{y}},
\textbf{\text{y}}_0)\\&=d_{\text{H}}(\textbf{\text{x}},
\chi(\textbf{\text{x}}))+d_{\text{H}}(\textbf{\text{y}},
\chi(\textbf{\text{y}}))\end{align*} and by construction of
$\chi'$,
$$d_{\text{H}}(\textbf{\text{x}}',
\chi'(\textbf{\text{x}}'))=d_{\text{H}}(\textbf{\text{x}}',
\chi(\textbf{\text{x}}')),\forall
\textbf{\text{x}}'\in\textbf{\text{X}}_1\backslash\{\textbf{\text{x}},
\textbf{\text{y}}\}.$$ By \eqref{eq-def-dst}, we have
$$d_{\text{S}}(\textbf{\text{X}}_1,\textbf{\text{X}}_2)=
d_{\chi}(\textbf{\text{X}}_1,\textbf{\text{X}}_2)=
d_{\chi'}(\textbf{\text{X}}_1,\textbf{\text{X}}_2).$$ Again by
construction of $\chi'$, we have $\mathcal N(\chi')=\mathcal
N(\chi)\backslash\{\textbf{\text{y}}\},$ and hence $$|\mathcal
N(\chi')|=|\mathcal N(\chi)|-1,$$ where
$$\mathcal
N(\chi')=\{\textbf{\text{y}}\in\textbf{\text{X}}_1\cap\textbf{\text{X}}_2;
\chi'(\textbf{\text{y}})\neq\textbf{\text{y}}\}.$$

If $\mathcal N(\chi')=\emptyset$, then
$\chi'(\textbf{\text{x}})=\textbf{\text{x}}$ for all
$\textbf{\text{x}}\in\textbf{\text{X}}_1\cap\textbf{\text{X}}_2$
and we can choose $\chi_0=\chi'$. Otherwise, by the same
discussion, we can obtain a
$\chi'':\textbf{\text{X}}_1\rightarrow\textbf{\text{X}}_2$ such
that $d_{\text{S}}(\textbf{\text{X}}_1,\textbf{\text{X}}_2)=
d_{\chi''}(\textbf{\text{X}}_1,\textbf{\text{X}}_2)$ and
$|\mathcal N(\chi'')|=|\mathcal N(\chi')|-1$, and so on. Noting
that $\mathcal
N(\chi'')\subseteq\textbf{\text{X}}_1\cap\textbf{\text{X}}_2$ is a
finite set, we can always find a $\chi_0\in\mathscr X$ such that
$d_{\text{S}}(\textbf{\text{X}}_1,\textbf{\text{X}}_2)=
d_{\chi_0}(\textbf{\text{X}}_1,\textbf{\text{X}}_2)$ and
$$\mathcal
N(\chi_0)=\{\textbf{\text{y}}\in\textbf{\text{X}}_1\cap\textbf{\text{X}}_2;
\chi_0(\textbf{\text{y}})\neq\textbf{\text{y}}\}=\emptyset.$$
Hence, we have $\chi_0(\textbf{\text{x}})=\textbf{\text{x}}$ for
all
$\textbf{\text{x}}\in\textbf{\text{X}}_1\cap\textbf{\text{X}}_2$,
which completes the proof.

\section{Proof of Lemma \ref{lem-dst-subset}}

It suffices to prove that if $\textbf{X}_2'\subseteq \textbf{X}_2$
and $|\textbf{X}_1|\leq|\textbf{X}_2'|=|\textbf{X}_2|-1$, then
$$d_{\text{S}}(\textbf{\text{X}}_1,\textbf{\text{X}}_2')\leq
d_{\text{S}}(\textbf{\text{X}}_1,\textbf{\text{X}}_2).$$

Without loss of generality, we can assume
\begin{align*}\textbf{X}_1&=\{\textbf{x}_1,
\cdots,\textbf{x}_n\}, \\
\textbf{X}_2'&=\{\textbf{y}_1,\cdots,\textbf{y}_n,
\textbf{y}_{n+1}, \cdots, \textbf{y}_{n+s-1}\}\end{align*} and
\begin{align*}
~~~~~~~~\textbf{X}_2=\{\textbf{y}_1,\cdots,\textbf{y}_n,
\textbf{y}_{n+1}, \cdots, \textbf{y}_{n+s-1}, \textbf{y}_{n+s}\},
\end{align*}
where $s\geq 1$, such that
$$d_{\text{S}}(\textbf{\text{X}}_1,\textbf{\text{X}}_2')
=\sum_{i=1}^nd_{\text{H}}(\textbf{\text{x}}_i,\textbf{\text{y}}_i)+L(s-1).$$
By Definition \ref{def-dst}, we can suppose
$$d_{\text{S}}(\textbf{\text{X}}_1,\textbf{\text{X}}_2)
=\sum_{i=1}^nd_{\text{H}}(\textbf{\text{x}}_i,\textbf{\text{y}}_{\ell_i})+Ls,$$
where $\{\ell_i;i=1,2,\cdots,n\}$ is a subset of
$\{1,2,\cdots,n+s\}$. We have the following two cases.

Case 1: $n+s\notin\{\ell_1,\ell_2,\cdots,\ell_n\}$. In this case,
we have
\begin{align*}d_{\text{S}}(\textbf{\text{X}}_1,\textbf{\text{X}}_2')
&=\sum_{i=1}^nd_{\text{H}}(\textbf{\text{x}}_i,\textbf{\text{y}}_{i})+L(s-1)\\
&\leq\sum_{i=1}^nd_{\text{H}}(\textbf{\text{x}}_i,\textbf{\text{y}}_{\ell_i})
+L(s-1) \\
&<\sum_{i=1}^n(d_{\text{H}}(\textbf{\text{x}}_i,\textbf{\text{y}}_{\ell_i})
+Ls\\
&=d_{\text{S}}(\textbf{\text{X}}_1,\textbf{\text{X}}_2),\end{align*}
where the first inequality is obtained by \eqref{eq-def-dst}.

Case 2: There exists a $k\in\{1,2,\cdots,n\}$ such that
$n+s=\ell_{k}$. Noticing that $s\geq 1$, then there exists an
$m\in\{1,2,\cdots,n+s-1\}$ such that
$m\notin\{\ell_1,\ell_2,\cdots,\ell_n\}$. Denote $\ell'_k=m$ and
$\ell'_i=\ell_i$ for $i\in\{1,2,\cdots,n\}\backslash\{k\}$. Then
we have \begin{align}\label{eq2-pf-dst-subset}\{\ell'_1, \ell'_2,
\cdots, \ell'_n\}\subseteq\{1,2,\cdots, n+s-1\}.\end{align}

Moreover, noticing that
$\{\textbf{x}_{k},\textbf{y}_{m},\textbf{y}_{\ell_k}\}\subseteq
\textbf{\text{X}}_1\cup\textbf{\text{X}}_2\subseteq\mathbb A^L$,
then $d_{\text{H}}(\textbf{x}_{k},\textbf{y}_{m})\leq L$ and
$d_{\text{H}}(\textbf{x}_{k},\textbf{y}_{\ell_k})\leq L$. Hence,
we can obtain
\begin{align}\label{eq-pf-dst-subset}
d_{\text{H}}(\textbf{x}_{k},\textbf{y}_{m})-
d_{\text{H}}(\textbf{x}_{k},\textbf{y}_{\ell_k})\leq L.
\end{align}
And further we have
\begin{align*}d_{\text{S}}(\textbf{\text{X}}_1,\textbf{\text{X}}_2')
&=\sum_{i=1}^nd_{\text{H}}(\textbf{\text{x}}_i,\textbf{\text{y}}_{i})+L(s-1)\\
&\leq\sum_{i=1}^nd_{\text{H}}(\textbf{\text{x}}_i,\textbf{\text{y}}_{\ell'_i})
+L(s-1)\\
&=\sum_{i=1}^nd_{\text{H}}(\textbf{\text{x}}_i,\textbf{\text{y}}_{\ell_i})
-d_{\text{H}}(\textbf{x}_{k},\textbf{y}_{\ell_k})
+d_{\text{H}}(\textbf{x}_{k},\textbf{y}_{m})\\&~~~ +L(s-1) \\
&\leq\sum_{i=1}^n(d_{\text{H}}(\textbf{\text{x}}_i,\textbf{\text{y}}_{\ell_i})
+L+L(s-1)\\&=d_{\text{S}}(\textbf{\text{X}}_1,\textbf{\text{X}}_2),\end{align*}
where the first inequality is obtained by
\eqref{eq2-pf-dst-subset} and \eqref{eq-def-dst}, and the second
inequality is obtained by \eqref{eq-pf-dst-subset}.

Thus, we always have
$d_{\text{S}}(\textbf{\text{X}}_1,\textbf{\text{X}}_2')\leq
d_{\text{S}}(\textbf{\text{X}}_1,\textbf{\text{X}}_2)$, which
completes the proof.

\section{Proof of Theorem \ref{dst-mtrc}}

By Definition \ref{def-dst}, it is easy to see that for any two
subsets $\textbf{\text{X}}_1$ and $\textbf{\text{X}}_2$ of
$\mathbb A^L$,
$d_{\text{S}}(\textbf{\text{X}}_1,\textbf{\text{X}}_2)
=d_{\text{S}}(\textbf{\text{X}}_2,\textbf{\text{X}}_1)\geq 0$.
Moreover, by Corollary \ref{cor-dst}, we can easily see that
$d_{\text{S}}(\textbf{\text{X}}_1,\textbf{\text{X}}_2)=0$ if and
only if $\textbf{\text{X}}_1=\textbf{\text{X}}_2$. To prove that
$d_{\text{S}}(\cdot,\cdot)$ is a distance function, we only need
to prove the triangle inequality, that is,
$$d_{\text{S}}(\textbf{\text{X}}_1,\textbf{\text{X}}_2)\leq
d_{\text{S}}(\textbf{\text{X}}_1,\textbf{\text{X}}_3)+
d_{\text{S}}(\textbf{\text{X}}_2,\textbf{\text{X}}_3)$$ for any
three subsets $\textbf{\text{X}}_1$, $\textbf{\text{X}}_2$ and
$\textbf{\text{X}}_3$ of $\mathbb A^L$. Without loss of
generality, we can assume that $|\textbf{\text{X}}_1|\leq
|\textbf{\text{X}}_2|$. Then we have the following three cases.

\textbf{Case 1}. $|\textbf{\text{X}}_1|\leq
|\textbf{\text{X}}_2|\leq|\textbf{\text{X}}_3|$. In this case, we
can fix a subset $\textbf{X}_3'\subseteq \textbf{X}_3$ of size
$|\textbf{X}_3'|=|\textbf{X}_2|$. Then by Lemma
\ref{lem-dst-subset},
$d_{\text{S}}(\textbf{\text{X}}_1,\textbf{\text{X}}_3')\leq
d_{\text{S}}(\textbf{\text{X}}_1,\textbf{\text{X}}_3)$ and
$d_{\text{S}}(\textbf{\text{X}}_2,\textbf{\text{X}}_3')\leq
d_{\text{S}}(\textbf{\text{X}}_2,\textbf{\text{X}}_3)$. It
suffices to prove that
$$d_{\text{S}}(\textbf{\text{X}}_1,\textbf{\text{X}}_2)\leq
d_{\text{S}}(\textbf{\text{X}}_1,\textbf{\text{X}}_3')+
d_{\text{S}}(\textbf{\text{X}}_2,\textbf{\text{X}}_3').$$ Without
loss of generality, we can assume
\begin{align*}\textbf{\text{X}}_1&=\{\textbf{\text{x}}_1,
\cdots,\textbf{\text{x}}_n\},\\
\textbf{\text{X}}_2&=\{\textbf{\text{y}}_1,\cdots,\textbf{\text{y}}_n,
\textbf{\text{y}}_{n+1},\cdots,\textbf{\text{y}}_{n+s}\},\\
\textbf{\text{X}}_3'&=\{\textbf{\text{z}}_1,\cdots,\textbf{\text{z}}_n,
\textbf{\text{z}}_{n+1},\cdots,\textbf{\text{z}}_{n+s}\}
\end{align*}
such that
$$d_{\text{S}}(\textbf{\text{X}}_1,\textbf{\text{X}}_3')
=\sum_{i=1}^nd_{\text{H}}(\textbf{\text{x}}_i,\textbf{\text{z}}_i)+Ls,$$
$$d_{\text{S}}(\textbf{\text{X}}_2,\textbf{\text{X}}_3')
=\sum_{i=1}^{n+s}d_{\text{H}}(\textbf{\text{y}}_i,\textbf{\text{z}}_i)$$
and
$$d_{\text{S}}(\textbf{\text{X}}_1,\textbf{\text{X}}_2)
=\sum_{i=1}^nd_{\text{H}}(\textbf{\text{x}}_i,\textbf{\text{y}}_{\ell_i})+Ls,$$
where $s\geq 0$ and
$\{\ell_1,\ell_2,\cdots,\ell_n\}\subseteq\{1,2,\cdots,n+s\}$. Then
we have
\begin{align*}d_{\text{S}}(\textbf{\text{X}}_1,\textbf{\text{X}}_2)
&=\sum_{i=1}^nd_{\text{H}}(\textbf{\text{x}}_i,\textbf{\text{y}}_{\ell_i})+Ls\\
&\leq\sum_{i=1}^nd_{\text{H}}(\textbf{\text{x}}_i,\textbf{\text{y}}_{i})+Ls \\
&\leq\sum_{i=1}^n(d_{\text{H}}(\textbf{\text{x}}_i,\textbf{\text{z}}_{i})
+d_{\text{H}}(\textbf{\text{y}}_i,\textbf{\text{z}}_{i}))+Ls\\
&\leq\sum_{i=1}^nd_{\text{H}}(\textbf{\text{x}}_i,
\textbf{\text{z}}_i)\!+\!Ls
\!+\!\sum_{i=1}^{n+s}d_{\text{H}}(\textbf{\text{y}}_i,
\textbf{\text{z}}_i)\\
&=d_{\text{S}}(\textbf{\text{X}}_1,\textbf{\text{X}}_3')
+d_{\text{S}}(\textbf{\text{X}}_2,\textbf{\text{X}}_3')\\
&\leq d_{\text{S}}(\textbf{\text{X}}_1,\textbf{\text{X}}_3)
+d_{\text{S}}(\textbf{\text{X}}_2,\textbf{\text{X}}_3),\end{align*}
where the first inequality is obtained by \eqref{eq-def-dst} and
the last inequality is obtained by Lemma \ref{lem-dst-subset}.

\textbf{Case 2}. $|\textbf{\text{X}}_1|\leq
|\textbf{\text{X}}_3|\leq|\textbf{\text{X}}_2|$. In this case, we
can assume
\begin{align*}\textbf{\text{X}}_1&=\{\textbf{\text{x}}_1,
\cdots,\textbf{\text{x}}_n\},\\
\textbf{\text{X}}_3&=\{\textbf{\text{y}}_1,\cdots,\textbf{\text{y}}_n,
\textbf{\text{y}}_{n+1},\cdots,\textbf{\text{y}}_{n+s}\},\\
\textbf{\text{X}}_2&=\{\textbf{\text{z}}_1,\cdots,\textbf{\text{z}}_n,
\textbf{\text{z}}_{n+1},\cdots,\textbf{\text{z}}_{n+s},
\textbf{\text{z}}_{n+s+1},\cdots,\textbf{\text{z}}_{n+s+t}\}
\end{align*}
such that
$$d_{\text{S}}(\textbf{\text{X}}_1,\textbf{\text{X}}_3)
=\sum_{i=1}^nd_{\text{H}}(\textbf{\text{x}}_i,\textbf{\text{y}}_i)+Ls,$$
$$d_{\text{S}}(\textbf{\text{X}}_2,\textbf{\text{X}}_3)
=\sum_{i=1}^{n+s}d_{\text{H}}(\textbf{\text{y}}_i,\textbf{\text{z}}_i)+Lt$$
and
$$d_{\text{S}}(\textbf{\text{X}}_1,\textbf{\text{X}}_2)
=\sum_{i=1}^nd_{\text{H}}(\textbf{\text{x}}_i,
\textbf{\text{z}}_{\ell_i})+L(s+t),$$ where $s,t\geq 0$ and
$\{\ell_1,\ell_2,\cdots,\ell_n\}\subseteq\{1,2,\cdots,n+s+t\}$.
Then we have
\begin{align*}d_{\text{S}}(\textbf{\text{X}}_1,\textbf{\text{X}}_2)
&=\sum_{i=1}^nd_{\text{H}}(\textbf{\text{x}}_i,
\textbf{\text{z}}_{\ell_i})+L(s+t)\\
&\leq\sum_{i=1}^nd_{\text{H}}(\textbf{\text{x}}_i,
\textbf{\text{z}}_{i})+L(s+t)\\
&\leq\sum_{i=1}^n(d_{\text{H}}(\textbf{\text{x}}_i,
\textbf{\text{y}}_{i})
+d_{\text{H}}(\textbf{\text{y}}_i,\textbf{\text{z}}_{i}))+L(s+t)\\
&\leq\sum_{i=1}^nd_{\text{H}}(\textbf{\text{x}}_i,
\textbf{\text{y}}_i)\!+\!Ls
\!+\!\sum_{i=1}^{n+s}d_{\text{H}}(\textbf{\text{y}}_i,
\textbf{\text{z}}_i)\!+\!Lt\\
&=d_{\text{S}}(\textbf{\text{X}}_1,\textbf{\text{X}}_3)
+d_{\text{S}}(\textbf{\text{X}}_2,\textbf{\text{X}}_3),\end{align*}
where the first inequality is obtained by \eqref{eq-def-dst}.

\textbf{Case 3}. $|\textbf{\text{X}}_3|\leq
|\textbf{\text{X}}_1|\leq|\textbf{\text{X}}_2|$. In this case, we
can assume
\begin{align*}\textbf{\text{X}}_3&=\{\textbf{\text{x}}_1,
\cdots,\textbf{\text{x}}_n\},\\
\textbf{\text{X}}_1&=\{\textbf{\text{y}}_1,\cdots,\textbf{\text{y}}_n,
\textbf{\text{y}}_{n+1},\cdots,\textbf{\text{y}}_{n+s}\},\\
\textbf{\text{X}}_2&=\{\textbf{\text{z}}_1,\cdots,\textbf{\text{z}}_n,
\textbf{\text{z}}_{n+1},\cdots,\textbf{\text{z}}_{n+s},
\textbf{\text{z}}_{n+s+1},\cdots,\textbf{\text{z}}_{n+s+t}\}
\end{align*}
such that
$$d_{\text{S}}(\textbf{\text{X}}_1,\textbf{\text{X}}_3)
=\sum_{i=1}^nd_{\text{H}}(\textbf{\text{x}}_i,\textbf{\text{y}}_i)+Ls,$$
$$d_{\text{S}}(\textbf{\text{X}}_2,\textbf{\text{X}}_3)
=\sum_{i=1}^{n}d_{\text{H}}(\textbf{\text{x}}_i,
\textbf{\text{z}}_i)+L(s+t)$$ and
$$d_{\text{S}}(\textbf{\text{X}}_1,\textbf{\text{X}}_2)
=\sum_{i=1}^{n+s}d_{\text{H}}(\textbf{\text{y}}_i,
\textbf{\text{z}}_{\ell_i})+Lt,$$ where $s,t\geq 0$ and
$\{\ell_1,\ell_2,\cdots,\ell_n\}\subseteq\{1,2,\cdots,n+s+t\}$.
Then we have
\begin{align*}d_{\text{S}}(\textbf{\text{X}}_1,\textbf{\text{X}}_2)
&=\sum_{i=1}^{n+s}d_{\text{H}}(\textbf{\text{y}}_i,
\textbf{\text{z}}_{\ell_i})+Lt\\
&\leq\sum_{i=1}^{n+s}d_{\text{H}}(\textbf{\text{y}}_i,
\textbf{\text{z}}_{i})+Lt\\
&\leq\sum_{i=1}^{n+s}(d_{\text{H}}(\textbf{\text{x}}_i,
\textbf{\text{y}}_{i})
+d_{\text{H}}(\textbf{\text{x}}_i,\textbf{\text{z}}_{i}))+Lt\\
&\leq\sum_{i=1}^nd_{\text{H}}(\textbf{\text{x}}_i,
\textbf{\text{y}}_i)\!+\!Ls
\!+\!\sum_{i=1}^{n}d_{\text{H}}(\textbf{\text{x}}_i,
\textbf{\text{z}}_i)\!+\!L(s\!+\!t)\\
&=d_{\text{S}}(\textbf{\text{X}}_1,\textbf{\text{X}}_3)
+d_{\text{S}}(\textbf{\text{X}}_2,\textbf{\text{X}}_3),\end{align*}
where the first inequality is obtained by \eqref{eq-def-dst}, and
the third inequality is obtained from the simple fact that
$d_{\text{H}}(\cdot,\cdot)\leq L$.

For all cases, we have proved that
$$d_{\text{S}}(\textbf{\text{X}}_1,\textbf{\text{X}}_2)\leq
d_{\text{S}}(\textbf{\text{X}}_1,\textbf{\text{X}}_3)+
d_{\text{S}}(\textbf{\text{X}}_2,\textbf{\text{X}}_3).$$ Hence,
$d_{\text{S}}(\cdot,\cdot)$ satisfies the triangle inequality.

By the above discussion, we proved that
$d_{\text{S}}(\cdot,\cdot)$ is a distance function over $\mathcal
P(\mathbb A^L)$.







\begin{thebibliography}{1} 

\bibitem{Davis96}
J. Davis, ``Microvenus,''  \emph{Art J}, 55, 70 (1996),
doi:10.2307/777811

\bibitem{Church12}
G. M. Church, Y. Gao, and S. Kosuri, ``Next-generation digital
information storage in DNA,'' \emph{Science}, vol. 337, no. 6102,
pp. 1628-1628, 2012.

\bibitem{Organick18}
L. Organick, S. D. Ang, Y. J. Chen, R. Lopez, S. Yekhanin, K.
Makarychev, M. Z. Racz, G. Kamath, P. Gopalan, B. Nguyen, C.
Takahashi, S. Newman, H. Y. Parker, C. Rashtchian, G. G. K.
Stewart, R. Carlson, J. Mulligan, D. Carmean, G. Seelig, L. Ceze,
and K. Strauss, ``Random access in large-scale DNA data storage,''
\emph{Nature Biotechnology}, vol. 36, pp. 242-248, Feb. 2018.

\bibitem{Rashtchian17}
C. Rashtchian, K. Makarychev, M. Racz, S. Ang, D. Jevdjic, S.
Yekhanin, L. Ceze, and K. Strauss, ``Clustering billions of reads
for DNA data storage,'' \emph{NIPS}, 2017.

\bibitem{Levenshtein01}
V. I. Levenshtein, ``Efficient reconstruction of sequences,''
\emph{IEEE Trans. on Inform. Theory}, vol. 47, no. 1, pp. 2-22,
Jan. 2001.

\bibitem{Goldman13}
N. Goldman, P. Bertone, S. Chen, C. Dessimoz, E. M. LeProust, B.
Sipos, and E. Birney, ``Towards practical, high-capacity,
lowmaintenance information storage in synthesized DNA,''
\emph{Nature}, vol. 494, no. 7435, pp. 77-80, 2013.

\bibitem{Grass15}
R. N. Grass, R. Heckel, M. Puddu, D. Paunescu, and W. J. Stark,
``Robust chemical preservation of digital information on DNA in
silica with error-correcting codes,'' \emph{Angew. Chem. Int.
Ed.}, vol. 54, no. 8, pp. 2552-2555, 2015.

\bibitem{Blawat16}
M. Blawat, K. Gaedke1, I. H$\ddot{\text{u}}$tter, X.-M. Chen, B.
Turczyk, S. Inverso, B. W. Pruitt, G. M. Church, ``Forward error
correction for DNA data storage,'' \emph{Procedia Compu Sci}, vol.
80, pp. 1011-1022, 2016.

\bibitem{Yazdi15}
S. M. H. T. Yazdi, Y. Yuan, J. Ma, H. Zhao, and O. Milenkovic, ``A
Rewritable, Random-Access DNA-Based Storage System,'' \emph{Nature
Scientific Reports}, 5(14138), 2015.

\bibitem{Bornholt16}
J. Bornholt, R. Lopez, D. M. Carmean, L. Ceze, G. Seelig, and K.
Strauss, ``A DNA-based archival storage system,'' in
\emph{Proceedings of the Twenty-First International Conference on
Architectural Support for Programming Languages and Operating
Systems, ACM}, pp. 637-649, 2016.

\bibitem{Erlich17}
Y. Erlich, and D. Zielinski, ``DNA Fountain enables a robust and
efficient storage architecture,'' \emph{Science}, vol. 355, no.
6328, pp. 950-954, 2017.

\bibitem{Wentu18}
W. Song, K. Cai, M. Zhang, and C. Yuen , ``Codes with Run-Length
and GC-Content Constraints for DNA-based Data Storage,''
\emph{IEEE Communications Letters}, 2018, DOI:
10.1109/LCOMM.2018.2866566

\bibitem{Immink18}
K.A.S. Immink, and K. Cai, ``Design of Capacity-Approaching
Constrained Codes for DNA-Based Storage Systems,'' \emph{IEEE
Communications Letters}, vol. 22, no. 2, pp. 224-227, 2018.

\bibitem{Lenz18}
A. Lenz, P. H. Siegel, A. W-Zeh, and E. Yaakobi, ``Coding over
Sets for DNA Storage,'' in \emph{Proc. IEEE Int. Symp. Inform.
Theory (ISIT)}, 2019, pp. 2411-2415.

\bibitem{Sima18}
J. Sima, N. Raviv, and J. Bruck, ``On Coding over Sliced
Information,'' 2018, Available: https://arxiv.org/abs/1809.02716

\bibitem{Heckel}
R. Heckel, I. Shomorony, K. Ramchandran, and D. N. C. Tse,
``Fundamental limits of DNA storage systems,'' in \emph{IEEE Int.
Symp. Inform. Theory (ISIT)}, Aachen, Germany, Jun. 2017, pp.
3130-3134.

\bibitem{Kiah-IT16}
H. M. Kiah, G. J. Puleo, and O. Milenkovic, ``Codes for DNA
sequence profiles,'' \emph{IEEE Trans. Inf. Theory}, vol. 62, no.
6, pp. 3125-3146, Jun. 2016.

\bibitem{Langberg-IT}
M. Langberg, M. Schwartz, and E. Yaakobi, ``Coding for the
$\ell_\infty$-Limited Permutation Channel,'' \emph{IEEE Trans.
Inf. Theory}, vol. 63, no. 12, pp. 7676-7686, Dec. 2017.

\bibitem{Kovacevic17}
M. Kova$\check{\text{c}}$evi$\acute{\text{c}}$, and V. Y. F. Tan,
``Codes in the Space of Multisets --- Coding for Permutation
Channels with Impairments,'' \emph{IEEE Trans. Inf. Theory}, 2018,
DOI: 10.1109/TIT.2017.2789292


\bibitem{Plotkin}
M. Plotkin, ``Binary codes with specified minimum distances,''
\emph{IEEE Trans. Inf. Theory}, vol. IT-6, pp. 445-450, 1960.

\bibitem{Johnson}
S. M. Johnson, ``A new upper bound for error-correcting codes,''
\emph{IEEE Trans. Inf. Theory}, vol. IT-8, pp. 203-207, 1962.

\bibitem{Munkres}
J.Munkres, ``Algorithms for the assignment and transportation
problems,'' \emph{J. Soc. Ind. Appl. Math.}, vol. 5, no. 1, pp.
32-38, Mar. 1957.

\bibitem{Huffman}
W. C. Huffman and V. Pless, \emph{Fundamentals of Error-Correcting
Codes}. Cambridge, U.K.: Cambridge Univ. Press, 2003

\bibitem{Jukna}
S. Jukna, \emph{Extremal Combinatorics}. New York:
Springer-Verlag, 2001.

\bibitem{Grassl}
M. Grassl, ``Bounds on the minimum distance of linear codes and
quantum codes,'' Online available at http://www.codetables.de.
Accessed on 2019-11-13.

\end{thebibliography}
\end{document}